\documentclass[12pt]{amsart}

  \usepackage{fullpage}
    \usepackage{amssymb}
 \usepackage[mathscr]{euscript}
\usepackage[all]{xy}

\newcommand\cL{\mathcal{L}}
\newcommand\cS{\mathcal{S}}

\newcommand\cT{\mathcal{T}}
\newcommand\cA{\mathcal{A}}
\newcommand\cB{\mathcal{B}}
\newcommand\cH{\mathcal{H}}
\newcommand\cD{\mathcal{D}}
\newcommand\cF{\mathcal{F}}
\newcommand\cX{\mathcal{X}}
\newcommand\cJ{\mathcal{J}}
\newcommand\cU{\mathcal{U}}

\newcommand\J{\mathbb{J}}
\newcommand\C{\mathbb{C}}
\newcommand\R{\mathbb{R}}
\newcommand\N{\mathbb{N}}
\newcommand\Z{\mathbb{Z}}
\newcommand\X{\mathbb{X}}
\newcommand\U{\mathbb{U}}
\newcommand\V{\mathbb{V}}

\newtheorem{thm}{Theorem}[section]
\newtheorem{prop}[thm]{Proposition} 
\newtheorem{cor}[thm]{Corollary}
\theoremstyle{remark}

\newtheorem{defn}[thm]{Definition}
\newtheorem{rem}{Remark}[section]
\newtheorem*{remn}{Remark}

\newtheorem*{proofn}{Proof}

\newcommand\dy{\mathop{dy}\nolimits}
\newcommand\dz{\mathop{dz}\nolimits}
\newcommand\cprime{$'$}
\newcommand\KP{\mathrm{KP}}
\newcommand\bj{{}^\beta\mathbb{J}}
\newcommand\bJ{\;^\beta\!{J}}
\newcommand\gd{Gel{\cprime}fand-Dickey}
\newcommand\vp{\varphi}
\newcommand\la{\lambda}
\newcommand\om{\omega}
\newcommand\olgc{\overline{\Gamma_c}}
\newcommand\Div{\mathop{\rm Div}\nolimits}
\newcommand\Tr{\mathop{\rm Tr}\nolimits}
\newcommand\vol{\mathop{\rm vol}\nolimits}
\newcommand\Res{\mathop{\rm Res}\nolimits}
\newcommand\Jac{\mathop{\rm Jac}\nolimits}
\newcommand\bsl{\backslash}
\newcommand\ol[1]{\overline{#1}}
\newcommand\bcs{\begin{cases}}
\newcommand\ecs{\end{cases}}
\newcommand\ga{\gamma}
\newcommand\Ga{\Gamma}
\newcommand\diag{\mathop{\rm diag}\nolimits}
 \newcommand\Id{\mathop{\rm Id}\nolimits}
\newcommand\ad{\mathop{\rm ad}\nolimits}
\newcommand\gl{\mathrm{gl}}
\newcommand\La{\Lambda}
\newcommand\lrec{\langle \enspace,\enspace\rangle}
\renewcommand\part{\partial}
\newcommand\hensp[1]{\hbox{ #1 }}
\newcommand\dx{\mathop{dx}\nolimits}
\newcommand\tr{\mathop{\rm tr}\nolimits}
\newcommand\Ad{\mathrm{Ad}}
\newcommand\eps{\epsilon}
\newcommand\vareps{\varepsilon}
\newcommand\nn{\nonumber}
\newcommand\eqx{\stackrel{(\ast)}{=}}

\newcommand\gbmp{g^\bot_{\mp}}
\newcommand\gbpm{g^\bot_\pm}

\newcommand\gbp{{g}^\bot_+}
\newcommand\gbm{{g}^\bot_-}
\newcommand\aks{AKS Theorem}
\let\lab\label
\let\pref\eqref
\newcommand\lra[1]{\left\langle#1\right\rangle}
\newcommand\lrp[1]{\left(#1\right)}
\newcommand\lrb[1]{\left[#1\right]}
\newcommand\ensp{\enspace}
\newcommand\lrc[1]{\left\{#1\right\}}

 \numberwithin{equation}{section}
\begin{document}
\title[The AKS Theorem, A.C.I. Systems and Random Matrix Theory]{The AKS Theorem, A.C.I. Systems\\ and Random Matrix Theory}
\author{Mark Adler}
\author{Pierre van Moerbeke}
\address{Department of Mathematics, Brandeis University, Waltham, MA 02454, USA}
\email{adler@brandeis.edu}
\address{Department of Mathematics, Universit\'e Catholique de Louvain, 1348 Louvain-la-Neuve, Belgium, and Department of Mathematics, Brandeis University, Waltham, MA 02454, USA}
\email{pierre.vanmoerbeke@uclouvain.be}
\thanks{The first named author  gratefully acknowledges the support of a Simons Foundation Grant \#27891.  The second named author gratefully acknowledges the support of a Simons Foundation Grant \#280945.}

\begin{abstract}
This paper gives the most general form of the Adler-Kostant-Symes (AKS) Theorem, and many applications of it, both finite and infinite dimensional, the former yielding algebraic completely integrable (a.c.i.) systems, and the latter examples in random matrix theory (RMT).
\end{abstract}
\maketitle

  \tableofcontents

\setcounter{section}{-1}

\section{Introduction}
The study of integrable systems begins with Euler's  top \cite{1} linearizing on an elliptic curve (thought of as an algebraic group), leading eventually to Liouville's Theorem \cite{2} concerning the connection between the existence of enough constants of the motion (integrals) in involution leading to a solution of the system  by quadrature, and the more refined Arnold-Liouville Theorem  leading to linearizing such systems on tori and the use of action-angle coordinates.  In practice, the tori were, at least in the canonical classical examples, algebraic tori and the motion was  actually linear motion on the (universal covering space of the) algebraic tori.  The existence of many integrals in involution was often seen to be a consequence of physical configurational symmetries, but certainly not always, the Kovalevsky top \cite{3} the best example where this was certainly not the case, so in effect, it was often a mystery.  Likewise the fact these systems linearized on algebraic tori was another mystery, starting from the case of the Euler top!

In the latter half of the twentieth century, much progress has been made, starting with the famous Fermi-Pasta-Ulam experiments and punctuated by the Lax formulation of the KdV equation \cite{4}.  Indeed the theory was broadened to new systems like the Toda lattice \cite{5}, which also was shown to have a Lax pair representation \cite{Flaschka}. Many other partial differential equations were introduced, like the the  KP  equation \cite{27}, as well as new methods of solution involving scattering theory, the Baker Akhiezer function \cite{6} and the $\tau$-function of Sato \cite{7} and his school.  As the subject matured it reached into new areas such as string theory, matrix models and random matrix theory (RMT) \cite{8}.

This paper will touch on a few of the above topics.  First the mystery of why systems have many constants of the motion in involution is well explained by the AKS Theorem \cite{1}, which says a Lie algebra  splitting as a vector space   into two  Lie algebras gives rise to integrals in involution and commuting vector fields whose solutions  linearize through group factorization (or splitting);  although the AKS Theorem does not say when systems are \emph{completely} integrable.  Of course Noether's Theorem \cite{9} is the first and more basic theorem, showing configurational group symmetries yield integrals.  In Section 1 we give the most general form of the \aks\ and in Sections 2 and 3 we give various examples.  In Section 2 almost all the examples are infinite dimensional examples  (like the KP equation), being  either lattices or PDEs, showing the robustness of the theorem; also these examples were picked as they all have application to RMT.

Then came the work on the periodic KdV equation \cite{McKPvM} and the periodic Toda lattice \cite{PvM}. Both systems were linearizable on Jacobians of hyperelliptic curves; see also \cite{KacvM0,KacvM}. These systems were extended to isospectral flows on difference operators \cite{PvMM}, also Lax pairs and Hamiltonian flows possessing many constants of the motion in involution; these flows linearize on Jacobians of more general set of curves. The Lax pair and the symplectic structure, depending on a free parameter $h$, gave rise to the spectral curve, together with the whole family of invariants in involution. These and other examples pointed in the direction of Kac-Moody Lie algebras.

Therefore Section 3 discusses examples coming from Kac-Moody Lie algebras, which moreover linearize on algebraic tori that are Jacobian of curves; they are examples of algebraic integrable systems (a.c.i.) \cite{10,1}.  General tools are discussed that are used to show when Lax-equations over Kac-Moody Lie algebras lead to systems linearizing on Jacobians of curves and hence a.c.i. systems \cite{11,12, 1}.  Sections 4, 5, and 6 deal with three examples from RMT \cite{13,15,14} having to do respectively with the KdV, Toda and Pfaff Lattice, and the two-Toda system.  The specialness of having coming from RMT also leads to the systems satisfying so-called Virasoro constraints, and the combination of the fact that the RMT examples can both be deformed to be integrable systems, solvable with $\tau$-functions, which moreover satisfy Virasoro constraints, leads to PDEs for the so-called gap probabilities associated with the random matrix systems.

The list of references is by no means exhaustive; apologies for this! References to 50 years of Toda Lattice-activity would require many more pages. More references can be found in the various papers mentioned.

\section{The AKS Theorem}

The \aks\ first  appeared in \cite{16}, then a more advanced version appeared in \cite{17} and later an $r$-matrix version appeared in \cite{18}.
Some basic elements of it appeared in \cite{19}.  Interestingly enough, the first version was just an attempt to get a Lie algebra theorem that simultaneously yielded both the Toda and KdV cases, which at the time people realized were deeply related.  The following version of the theorem is a synthesis of all the above, although not the most abstract version (but most useful), since we identify $g\cong g^\ast$, the Lie algebra with its dual, via an Ad-invariant non-degenerate bilinear form; it appears in \cite{1}.  The point of the theorem is that vector space decompositions of the Lie aglebras, where the components are subalgebras, lead to integrable systems.  

First some preliminaries are in order.  Assume we are given an Ad-invariant\footnote{$\lra{[X,Y], Z}= \lra{X, [Y,Z]}$.} non-degenerate bilinear form on a Lie algebra $g$,
\begin{equation}\lab{1}
\lra{\enspace,\enspace} g\times g\to \C,\end{equation}
which induces $\nabla F(X)\in g$ the gradient of $F$ at $X$ for functions $F$ on $g^\ast\cong g$ by
\begin{equation}\lab{2}
dF(X)=\lra{\nabla F(X),dX}\end{equation}
and the Kostant-Kirillov Poisson structure on $g^\ast\cong g$ with respect to $\lra{\cdot,\cdot}$ by
\begin{equation}
\lab{3} \{F,H\}(X)=\lra{X,[\nabla F(X),\nabla H(X)]}.\end{equation}
The isomorphism $g\to g^\ast$ induced by $\lra{\cdot,\cdot}$ established a one-to-one correspondence between Ad$^\ast$-invariant functions on $g^\ast$ and Ad-invariant functions on $g$, and so the Hamiltonian vector fields $\cX_H$ on $(g^\ast\cong g$, $\lrc{\enspace,\enspace})$ take the simple Lax form (which is why we identify $g^\ast\cong g$)
\begin{equation}\lab{4}
\cX_H(X)\equiv \dot X=[\nabla H(X),X].\end{equation}
Note that in general
\begin{equation}\lab{5}
\cX_H(F)= \lra{\nabla F,\cX_H}:=\{F,H\}.\end{equation}
Also note that the Casimirs  of $\cX_H$, i.e., the Hamiltonians that produce null vector-fields, and hence are constant along symplectic leaves, are precisely  the $\Ad^\ast$ invariant functions: $\cF(g^\ast)^G\cong \cF(g)$, characterized by $0=\cX_H(X)=[\nabla H(X),X],\; \forall X\in g$.
Given a \emph{vector space} splitting of the Lie algebra $g$ as follows,
\[
g=g_+ \oplus g_-\]
with $g_\pm$ Lie algebras, with the above a vector space direct sum (as opposed to a Lie algebra direct sum).  In view of the non-degeneracy of $\lra{\cdot,\cdot}$, we furthermore have the vector space direct sum and isomorphisms:
\begin{equation}\lab{6}
g^\ast = g\cong g^\bot_+ \oplus g^\bot_-, \hbox{ with }g^\bot_\pm \cong g^\ast_\mp,\end{equation}
where $g^\bot_\pm$ is the orthogonal complement (with respect to $\lra{\cdot,\cdot}$) of $g_\pm$, and therefore they carry the Kostant-Kirillov Poisson structure $\{\enspace,\enspace\}_{g^\bot_\pm \cong g^\ast_\mp}$, namely,
\begin{equation}\lab{7}
\cX_H(X)\big|_{g^\bot_\mp}=\hat P_\mp [\nabla_\pm H,X],\qquad X\in g^\bot_\mp \cong g^\ast_\pm\end{equation}
with
\[
dH(X)=\lra{dX,\nabla_\pm H},\qquad \nabla_\pm H\in g_\pm \quad \hbox{(since $dX\in g^\bot_\mp$}),\]
with $\hat P_\mp$ being projections onto $g^\bot_\mp$ along $g^\bot_\pm$.  Indeed, since by \pref{5} and \pref{3} and  the  Ad-invariance of $\lra{\ensp,\ensp}$, 
\begin{align*}
\lra{\cX_H(X)\big|_{g^\bot_\mp},\nabla_\pm F(X)} &=
\cX_{H}\big|_{\gbmp} (F(X)) = \{F,H\}\big|_{\gbmp} (X)\\
&=\lra{X,[\nabla_\pm F,\nabla_\pm H]} = \lra{[\nabla H_\pm,X],\nabla_\pm F}\\
&= \lra{\hat P_\mp [\nabla H_\pm,X],\nabla F_\pm},\end{align*}
we find \pref{7}.

We shall need the group analog of the splitting of $g$.  So denote by $G$ the (connected, simply connected) Lie group whose Lie algebra is $g$ and correspondingly by $G_\pm$ going with~$g_\pm$.  

Going with the decomposition $g=g_+\oplus g_-$ we have the projections $P_\pm g\to g_\pm$, and now define $R=P_+ - P_-$ and the new Lie algebra on $g$
\begin{equation}\begin{aligned}\lab{8}
[X,Y]_R&\eqx \frac12([RX,Y]+[X,RY])\\
&= [X_+,Y_+]-[X_-,Y_-], \end{aligned}\end{equation}
with
\begin{equation}\lab{9}
X_\pm = P_\pm X.\end{equation}
Indeed, given the first definition of $[\ensp,\ensp]_R$, namely $\eqx$, and a general linear map $R:g\to g$, then it yields a Lie algebra (satisfies the Jacobi-identity) provided
\[
[RX,RY]-R([RX,Y]+[X,RY])= -c[X,Y],\]
the modified Yang-Baxter equation ($c=0$ in the Yang-Baxter equation).  The case $R=P_+-P_-$ corresponds to $c=1$ (which by rescaling can always be assumed if $c\ne 0$), but in general if we set $c=1$ and assume $R$ satisfies $\eqx$, then
\[
g_\pm:= \{X\pm RX\mid X\in g\} \]
are always Lie subalgebras of $g=g_+\oplus g_-$ and then $R=P_+ -P_-$, with $P_\pm$ the projections onto the $g_\pm$.  The version of the \aks\ given below shall only work for the case $R=P_+-P_-$.  Below the subscript $(\ensp)_R$ shall indicate the structure induced by $[\ensp,\ensp]_R$.

The Kostant-Kirillov structure induced by  $[\ensp,\ensp]_R$ and $\lra{\ensp,\ensp}$ on $g^\ast\cong g$ when $R=P_+-P_-$ is given by (remember $\hat P_\mp$ are projections onto $\gbmp$ along $\gbpm$)
\begin{equation}
\lab{10}
\cX_H(X)\big|_R = \hat P_- [(\nabla H)_+,X]-\hat P_+[(\nabla H)_-,X],\end{equation}
since by \pref{5} and \pref{3}
\begin{align*}
&\lra{\cX_H(X)\big|_R,\nabla F} =\cX_H(F)=\{F,H\}_R\\
 ={}& \lra{X,[\nabla F,\nabla H]_R}  =  \lra{X,[(\nabla F)_+,(\nabla H)_+]-[(\nabla F)_-,(\nabla H)_-]}\\
 ={}& \lra{[(\nabla H)_+,X], (\nabla F)_+} -\lra{[(\nabla H)_-,X],\nabla F_-}\\
 ={}& \lra{\hat P_-[(\nabla H)_+,X]-\hat P_+[(\nabla H)_-,X],\nabla F}.\end{align*}
 With all the preliminaries out of the way, we can now state the \aks.
 
\begin{thm}[\aks\ on $g$] \lab{thm1.1}
Suppose that $g=g_+\oplus g_-$ is a Lie algebra splitting and that $\lra{\cdot,\cdot}$ is an $\Ad$-invariant non-degenerate bilinear form on $g$, leading to a vector space splitting
\begin{equation}
\lab{11}
g= \gbp \oplus \gbm \cong g^\ast_-\oplus g^\ast_+.\end{equation}
Let $F,H\in\cF(g)^G$ and suppose that $\eps\in g$ satisfies
\begin{equation}\lab{12}
[\eps,g_+]\in \gbp,\qquad [\eps,g_-]\in \gbm.\end{equation}
Setting  $F_\eps(x):=F(\eps+X)$, then 
{\setbox0\hbox{(1)}\leftmargini=\wd0 \advance\leftmargini\labelsep
 \begin{enumerate}
\item $\{F,H\}_R\!=\!0$ and $\{F_\eps,H_\eps\}_{\gbm}\!=\!0$; hence $[\cX_F,\cX_H]$ and $[\cX_{F_\eps},\cX_{H_\eps}]\!=\!0$.
\item The Hamiltonian vector fields $\cX_H := \{\cdot,H\}_R$ and $\cX_{H_\eps}:= \{\cdot,H_\eps\}_{\gbm}$ are respectively given by
\begin{equation}
\lab{13}
\cX_H (X)=-\frac12 [X,R(\nabla H(X))]=\pm [X,(\nabla H(X))_\mp]\end{equation}
and
\begin{equation}\lab{14}
\cX_{H_\eps}(X)=-\frac12 [Y,R(\nabla H(Y))] = \pm [Y,(\nabla H(Y))_\mp],\end{equation}
where $Y\in \gbm + \eps$, yielding two families of commuting vector fields.
\item For $X_0\in g$ and for $|t|$ small, let $g_+(t)$ and $g_-(t)$ denote the smooth curves in $G_+$ resp.\ $G_-$ that solve the factorization problem
\begin{equation}\lab{15}
\exp(-t \nabla H(X_0))=g_+(t)^- g_-(t),\qquad g_\pm (0)=e.\end{equation}
Then the integral curve of $\cX_H$ that starts at $X_0$ is given for $|t|$ small by
\begin{equation}\lab{16}
X(t)=\Ad_{g_+(t)} X_0 = \Ad_{g_-(t)}X_0.\end{equation}\end{enumerate}}  \end{thm}

\begin{rem}\lab{rem1}
Maybe the most amazing thing about the \aks\ is the diversity of examples that it covers.  The theorem should be understood to include the case where $g$ is infinite dimensional, where only the first part of (1), and (3) needs to be interpreted carefully (or formally), and we shall not dwell on such issues here.

\end{rem}

\begin{rem}\lab{rem2} The \aks\ is thought of as a theorem about integrability, but classically integrable systems were solved by integration (quadrature) while the systems here are solved in (3) by the algebraic process of factoring an element in $G$, i.e., solving the splitting problem $G=G_+G_-$, which can be done near their identity, and moreover through algebraic operations for finite dimensional groups or analytic operations for say Kac-Moody groups.
\end{rem}

\begin{rem}
Classically complete integrability means, say on a Poisson manifold, that in addition to the Casimirs that define the symplectic leaves and lead to null Hamiltonian vector fields, there are $\frac12$ (dimensions of the symplectic leaves) number of additional commuting integrals (generically independent).  In many cases of the \aks\ this is indeed the case, either for the generic symplectic leaves and/or degenerate symplectic leaves (like on the classical Toda system).  It would be nice to have some general theorem that gives useful hypotheses to ensure complete integrability, especially as $\cF(g)^G$ often does not provide enough integrals.  An interesting and beautiful such example is \cite{20}, where the action of a parabolic subgroup supplies additional integrals.
\end{rem}
\begin{rem}
It would be wonderful to quantize AKS in a meaningful way to actually yield in a consistent and uniform way quantized ``integrable" systems.
\end{rem}
\begin{rem}
In may examples, as we shall see, you need $\eps\ne 0$, although in the original Toda lattice, you may take $\eps=0$, even though $\eps\ne 0$ first appears in that case; hence it is important to include the case $\eps\ne 0$ in the theorem.  Secondly, even though we need $R=P_+-P_-$, the original case of the theorem, the $R$ matrix version of the theorem allows for more initial conditions $X\in g$ than the original case where $Y\in \gbm+\eps$, and that flexibility  is useful in many examples, as we shall see.  Another source of flexibility in the theorem is that any choice of $\Ad$-invariant inner product (i.e., not just the Killing form) will do, and that effects the flow, as it effects the form of $\nabla H$.\end{rem}

\begin{rem}
It is not always to say what examples cannot be covered by a theorem or some natural generalization of it, as for example, the Kowalevsky top was eventually shown in \cite{BPS} to be an example of the AKS theorem, after some nontrivial effort by a number of authors.  However, clearly the theorem in its present form does not cover any quantum integrable systems or systems related to supergroups, etc.  There are, for example, many systems in random matrices that satisfy: differential equations through the use of the Riemann-Hilbert method, and it is hard to claim that these systems would not eventually be covered by some integrability theorem, whether it be AKS or some other theorem to be discovered.  We do not understand integrability well enough to really answer such questions or even offer idle speculation. 
\end{rem}

\section{Examples of the AKS Theorem}

\subsection*{Example 1: The KP and Gel$'$fand Dickey hierarchies}\

This subsection deals with the KP hierarchy and its invariant subsystems, the Gel$'$fand-Dickey hierarchies \cite{21}.  The KP equation describes water waves in two space dimensions where the water feels the bottom of its container, such as waves at the beach or tsunami waves generated by an undersea earthquake through displacement of the ocean bottom, traveling across the ocean at approximately the speed of sound while maintaining their profile to a very high degree.  The KdV equation, the most famous of the Gel$'$fand-Dickey equations, deals with waves in one space dimension, such as you would find in a canal.  In general, the Gel$'$fand-Dickey equations come up in various areas of mathematical physics, like string theory, random matrix theory, etc., no doubt due to their complete integrability.

Consider the Lie algebra 
\begin{equation}\lab{2.1}
g=\lrc{\sum_{-\infty<i\ll \infty} a_i(x)D^i \mid a_i(x)\in R},\footnotemark \qquad D=\frac{d}{dx},\end{equation}
\footnotetext{$i\ll \infty$ means all $i$ are less than some finite $N$ for any element of $g_+$, but $N$ varies with the element.}%
the ring of pseudo-differential operators over $\R$ or $S^1$, with $R$ the ring of differential functions over $\R$ or $S^1$ and
\begin{align} \label{2.2}
g=g_+ \oplus g_-,\qquad g_+& = \lrc{\sum_{0\leq i\ll \infty} a_i(x)D^i  \mid a_i(x)\in R} = g^\bot_+ \cong g^\ast_-\\
g_-&= \lrc{\sum_{-\infty<i\leq -1} a_i(x)D^i\mid a_i(x)\in R} = g^\bot_-\cong g^\ast_+.\nn\end{align}
The latter equalities are a consequence of (the Adler trace \cite{16})
\begin{equation}\lab{2.3}
\lra{a,b}=\tr(ab),\quad \tr(\Sigma a_i(x)D^i) = D^{-1} a_{-1}(x):=\int a_{-1}(x)\dx,\end{equation}
and note that  $\tr[a,b]=0$, so $\lra{\ensp,\ensp}$ is Ad-invariant.

Also set $\eps=D$, which satisfies \eqref{12}.  The functions 
\begin{equation}\lab{2.4}
H^{(\ell)}(a)=\frac{\tr a^{\ell+1}}{\ell+1}\in\cF(g)^G,\end{equation}
since
\[
d H^{(\ell)}(a)=\lra{a^\ell,da},\quad \nabla H^{(\ell)}(a)=a^\ell\implies [\nabla H^{(\ell)} (a),a]=0,\]
and so the Hamiltonian vector fields of Theorem \ref{thm1.1}
\[
\cX_{H^{(\ell)}_\eps} := \{ \cdot,H^{(\ell)}_\eps\}_{g^\bot_-}\hensp{on}
Y= D+\sum^\infty_{i=1} a_i D^{-i}\in g^\bot_- +\eps,\]
given by \pref{14}, take the form
\begin{equation}\lab{2.5}
\frac{\part Y}{\part t_\ell}:= \cX_{H^{(\ell)}_\eps} (Y)=\pm [(Y^\ell)_\pm,Y],\end{equation}
which is nothing but the $\ell$-th flow of the KP hierarchy \cite{7}.

Hence the KP hierarchy is an ``integrable" Hamiltonian system, with Hamiltonian structure given by the Kostant-Kirillov symplectic structure on $g_- = g^\bot_-\cong g^\ast_+$, namely, \pref{7}. Indeed, setting $X=\sum^\infty_{i=1} a_i D^{-i}\in g_-$, the Hamiltonian vector field generated by $H(X)$ is given by:\footnote{Note that $D^{-n}(f D^r (g_\bullet)) = \sum^\infty_{j=0} {-n\choose j} (D^j f) D^{r-n-j}(g_\bullet)$.}
\begin{equation}\lab{2.6}
\cX_H(X)=P_- [\nabla_+ H,X],\qquad \nabla H_+(\bullet) = \sum^\infty_{i=1} D^{i-1} \lrp{\frac{DH}{Da_i}\bullet} .\end{equation}
 The latter using
 \[
 dH(X):= \sum^\infty_1 \frac{DH_i}{Da_i}da_i =: \lra{dX,\nabla_+H},\; \nabla_+ H\in g_+,dX=\sum^\infty_{i=1} da_i D^i,\]
 while the $\ell$-th flow of the hierarchy is generated by the Hamiltonian
 \begin{equation}\lab{2.7}
 H^{(\ell)}_D(X)=(\tr (D+X)^{\ell+1})/(\ell+1).\end{equation}
 Of course \pref{16} of the \aks\ is only a formal solution to the hierarchy.  As far as we know, the above result is new.
 
 Note that $Y^n\in g_+$ is invariant under the KP hierarchy, and then the flows are called the $n$-Gel$'$fand-Dickey hierarchy.  We can think of these as flows in their own right, living in $g_+ = g^\bot_+\simeq g^\ast_-$, so in the \aks\ we just interchange the role of $g_\pm$.  Now the induced Hamiltonian structure \pref{7} on $g^\bot_+\simeq g^\ast_-$ is given by
 \begin{equation}\lab{2.8}
 \cX_H(X)=P_+ [\nabla_- H,X],\; X=\sum_{0\leq i\ll \infty} a_i D^i,\; \nabla_- H = \sum^{\infty}_{i=1} D^{-i-1}
 \lrp{\frac{DH}{D a_i}\bullet}.\end{equation}
Note 
 \begin{equation}\lab{2.9}
 \cA_n = \lrc{ \sum^n_{i=0} a_i D^i,\hensp{with} a_n=1, a_{n-1}=0} \end{equation}
 is an invariant manifold of the (co-adjoint and hence the) Hamiltonian action on $g^\bot_+\simeq g^\ast_-$ given by $\cX_H$, i.e., a union of symplectic leaves, sometimes called a Poisson subspace.  So we may consider $\cX_{H|_{\cA_n}}$ (the Gel$'$fand-Dickey symplectic structure as described in \cite{16}) and setting $\eps=0$ in \pref{12} and
 \begin{equation}
 \lab{2.10}
 H^{(\ell)}_0(X)=\frac{n}{n+\ell} \tr X^{\frac{\ell+n}{n}},\; X\in \cA_n,\end{equation}
 yields by \pref{14}, with $(\ensp)_\mp\to (\ensp)_\pm$ in AKS the $\ell$-th flow of the Gel$'$fand-Dickey hierarchy \cite{21} 
 \begin{equation}\lab{2.11}
 \frac{\part X}{\part t_\ell}:=\cX_{H^{(\ell)}_0} (X)=\pm \left[X,\lrp{X^{\frac{\ell}{n}}}_\pm\right].\end{equation}
 Once again \pref{16} of the \aks\ provides only a formal solution.
 
 \subsection*{Example 2:  The Toda lattice on $g\ell(n+1)$}
 \
 
 The famous Toda lattice was invented by M. Toda \cite{5} as a simple model for a one-dimensional crystal in solid state physics with nearest neighbor interactions given by the exponential function.  Both the periodic and nonperiodic lattice, as well as various Lie algebra versions have proven to be important in mathematical physics and \cite{BIZ} enumerative geometry dealing with string theory.  The system is also important in orthogonal polynomial theory \cite{30} and random matrix theory \cite{28}.
 
 Here we set $g=g\ell (n+1)= g_+\oplus g_-$, where $g_+$ (respectively $g_-$) $=$ Lie algebra of lower triangular matrices (resp.\ Lie algebra of strictly upper triangular matrices), while $\lra{A,B}=\tr AB$; hence $g^\ast_-(\hbox{resp. }g^\ast_+)\cong g^\bot_+$ (resp.~$g^\bot_-$) $=$ Lie algebra of strictly lower triangular matrices (resp.\ upper triangular matrices), and $\eps=$ matrix with all 1's one below the main diagonal and all other entires are $0$.  Here $n$ may be infinite.

 Since the induced Hamiltonian structure on $\gbm\cong g^\ast_+$ is given by \pref{7} 
 \[
 \cX_H(X)=\hat P_-[\nabla_+ H,X],\]
 with $X\in\gbm$, $\hat P_-$ projection onto $\gbm$ ($g=\gbm\oplus \gbp$)
 and
 \[
 dH=\lra{dX,\nabla_+ H}, \; \nabla_+H\in g_+,\]
then
 \[
 B_p = \{ X\in \gbm\hensp{with at most $p$ bands$\}$}\]
 is an invariant manifold of the Hamiltonian action given by $\cX_H$, so we may consider $\cX_H |_{B_p}$.  If we set
 \[
 H^{(\ell)}_\eps (X) = -\frac{\tr(\eps+X)^{\ell+1}}{\ell+1},\quad Y=\eps+X,\]
 the equations \pref{14} of the \aks\   yield 
 \begin{equation} \lab{2.12}
 \frac{\part}{\part t_\ell} Y:= \cX_{H^{(\ell)}_\eps}(Y) = \pm [Y,(Y^\ell)_\pm],\qquad X\in B_2\end{equation}
(in the style of \cite{19}) 
the Toda hierarchy, with $\ell=1$ the classical Toda equation:
\[
\dot Y = [Y,Y_+]=-[Y,Y_-],\]
where
\[
Y:= \begin{pmatrix}
b_1&a_1&&&0\\
1&b_2&a_2\\
&\ddots&\ddots&\ddots\\
&&1&b_n&a_n\\
0&&&1&b_{n+1}\end{pmatrix}\]
and where
\[
Y_- = \begin{pmatrix}
0&a_1&&&0\\
0&0&a_2\\
&\ddots&\ddots&\ddots\\&&0&0&a_n
\\
0&&&0&0\end{pmatrix},\quad Y_+=\begin{pmatrix}
b_1&0&&&0\\
1&b_2&0\\
&\ddots&\ddots&\ddots\\
&&1&b_n&0\\
0&&&1&b_{n+1}\end{pmatrix}.\]

Formula \eqref{16} of AKS yields the explicity solution of the hierarchy.  It is worth noting that the Toda equations can be gotten from another decomposition in the style of \cite{16}, namely,
\begin{equation}\lab{2.13}
g=g\ell(n+1)=g_+\oplus g_-,\end{equation}
where
\begin{align}\label{2.14}
 g_+(\hbox{resp. }g_-) ={}& \text{ lower triangular matrices with diagonal}\\
 &\text{(resp.\ skew-symmetric matrices),}\nn\end{align}
with $\lra{\cdot,\cdot}$ the same, hence 
\begin{align}\label{2.15}
&g^\ast_-\hbox{ (resp.\ $g^\ast_+$)} \\ \notag
&\cong \gbp \hbox{ (resp. $\gbm$) } = \text{ lower triangular matrices
with}\\ &\hspace*{1.25in}\hbox{no diagonal (resp.\ symmetric matrices),}\nn\end{align}
and $\eps = 0$. The induced Hamiltonian structure \pref{7} on $\gbm$ is given by 
\begin{equation}\lab{2.16}
\cX_H(X)=\hat P_-[\nabla_+ H,X],\quad dH=: \lra{dX,\nabla_+ H},\quad \nabla_+ H\in g_+.\end{equation}
Observe
\begin{equation}\lab{2.17}
A_p = \{X\in \gbm\text{ with at most $p$ bonds above and below the diagonal}\}\end{equation}
is an invariant manifold of the Hamiltonian action, so we may consider $\cX_H\big|_{A_p}$.  If we set
\begin{equation}\lab{2.18}
H^{(\ell)}_0 (X)=-\frac{1}{2} \frac{\tr X^{\ell+1}}{\ell+1},\qquad X\in \gbm,\end{equation}
the equations \pref{14} of the \aks
\begin{equation}\lab{2.19}
\frac{\part X}{\part t_\ell}:= \cX_{H^{(\ell)}_0}(X)=\pm \frac12 [X,(X^\ell)_\pm],\qquad X\in \cA_2 \end{equation}
yield the Toda hierarchy with $\ell=1$ the classical Toda equations, but in slightly different coordinates than with the previous splitting.  Once again \pref{16} of AKS yields the explicit solution of the hierarchy.

\subsection*{Example 3: The two-Toda lattice}\

The two-Toda lattice, a deep generalization of Ueno and Tagasaki \cite{22} of the one-Toda lattice, occurs in various areas of mathematical physics, such as random matrix theory \cite{EM}, \cite{15}, Gromov-Witten theory and Hurwitz number \cite{OP}, as well as matrix models in string theory, to name a few examples.
We finally need the $R$-matrix version of  the \aks.

Consider the splitting of the algebra $g$ of pairs $(P_1,P_2)$ of infinite $(\Z\times \Z)$ or semi-infinite $(\N\times \N)$ matrices such that $(P_1)_{ij}=0$ for $j-i\gg 0$ and $(P_2)_{ij}=0$ for $i-j\gg 0$,\footnote{So $(P_1)_{ij}=0$ for $j>i$ eventually, i.e., when $j-i$ is sufficiently large, etc.}
 used in \cite{15}; to wit:
\begin{align}
\label{2.20}
g&=g_++ g_-,\\
g_+&= \lrc{(P,P)\mid P_{ij}=0\hensp{if} |i-j|\gg 0} = \lrc{(P_1,P_2)\in g\mid P_1=P_2},\nn
\\ g_- &= \lrc{(P_1,P_2)\mid (P_1)_{ij}=0\hensp{if}  j\geq i,\; (P_2)_{ij}=0\hensp{if} i>j},\nn\end{align}
with $(P_1,P_2)= (P_1,P_2)_+ + (P_1,P_2)_-$ given by
\begin{align*}
(P_1,P_2)_+ &= (P_{1u}+P_{2\ell},P_{1u}+P_{2\ell}),\\
(P_1,P_2)_-&= (P_{1\ell} - P_{2\ell},P_{2u}-P_{1u});\end{align*}
$P_u$ and $P_\ell$ denote the upper (including diagonal) and strictly lower triangular parts of the matrix $P$, respectively.  

Take for $\lra{\ensp,\ensp}$ on $g$, $\lrec_1+\lrec_2$, with $\lra{A,B}_i = \tr AB$ on the $i$-th components of $g$; i.e., $\lrec$ just decouples (as does the Lie bracket) so it is Ad-invariant.  Then let $L=(L_1,L_2)$ be the running variables on $g\cong g^\ast$ and consider the (formal) Hamiltonians
\begin{equation}\lab{2.21}
H^{(i)}_n(L)=\frac{\tr L^{n+1}_i}{n+1},\qquad i=1,2,\; n=1,2,\dots \;.\end{equation}
Then under the Hamiltonian vector-fields $\cX_{H^{(i)}_n(L)  |_R}$ we find
\begin{align}
\label{2.22}
\frac{\part L}{\part t_n} =:\cX_{H^{(1)}_n(L) |_R} (L)&=[(L^n_1,0)_+,L]\\
\nn \frac{\part L}{\part s_n} =: \cX_{H^{(2)}_n(L) |_R} (L)&= [(0,L^n_2)_+,L],\end{align}
which are deformations of a pair of infinite matrices
\begin{equation}\lab{2.23}
L=(L_1,L_2) = \Big( \sum_{-\infty<i\leq 1} a^{(1)}_i \La^i,\; \sum_{-1\le i<\infty}a^{(2)}_i \La^i\Big)\in g,\end{equation}
with $\La$ the shift operator and where $a^{(1)}_i$ and $a^{(2)}_i$ are diagonal matrices depending on $t=(t_1,t_2,\dots)$ and $s=(s_1,s_2,\dots)$, such that
\[
a^{(1)}_1=I\hensp{and} \lrp{a^{(2)}_{-1}}_{nn} \ne 0\hensp{for all} n;\]
that is to say matrices of the above form are an invariant manifold of the flows, and the flows restricted to this manifold are called the two-Toda flows \cite{22}.  While $\La=(\delta_{j=i+1})$, in the semi-infinite case we need to set $\La^{-i}=(\La^i) ^\top$ for $i\geq 1$, and once again, the flows restrict to \pref{2.23} in the semi-infinite case, yielding the semi-infinite two-Toda flows.  Of course the Hamiltonians may not converge,  in which case they are ``formal Hamiltonians," but the flows make perfectly good sense in any case, and they all commute.

\subsection*{Example 4: The Pfaff Lattice}\

The Pfaff lattice comes up as the natural integrable system that is the deformation class for the GOE and GSE examples in random matrix theory \cite{29}, just as the Toda system is the natural  integrable system with which to deform the GUE case of random matrix theory.  Similarly it is the natural deformation class of skew-orthogonal polynomials, just as the Toda system is the natural deformation class of orthogonal-polynomials.  There is an associated tau-function theory for this hierarchy and it fits in naturally to the other Sato hierarchies~\cite{27}.

First consider the Lie algebra $\cD=g\ell_\infty$ of semi-infinite matrices, viewed as being composed of $2\times 2$ blocks.  It admits the natural decomposition into subalgebras as follows~\cite{23}:
\begin{equation}\lab{2.24}
\cD=\cD_- \oplus \cD_0\oplus \cD_+ = \cD_- \oplus \cD^-_0 \oplus \cD^+_0\oplus \cD_+,\end{equation}
where $\cD_0$ has $2\times 2$ blocks along the diagonal with zeroes everywhere else, and where $\cD_+$ (resp.\ $\cD_-)$ is the subalgebra of upper-triangular (resp.\ lower-triangular) matrices with $2\times 2$ zero matrices along $\cD_0$ and zero below (resp.\ above).  As we point out in \pref{2.24}, $\cD_0$ can further be decomposed into two Lie subalgebras as follows:
\begin{align}\label{2.25}
\cD^-_0&=\{\text{all $2\times 2$ blocks $\in \cD_0$ are proportional to Id}\},\\
\nn \cD^+_0&=\{\text{all $2\times 2$ blocks $\in \cD_0$ have trace }0\}.\end{align}
Consider the following: the semi-infinite skew-symmetric matrix $J$, zero everywhere, except for the following $2\times 2$ blocks, along the ``diagonal,"
\begin{equation}\lab{2.26}
J=\begin{pmatrix}
0&1\\
-1&0\\
&&0&1\\
&&-1&0\\
&&&&0&1\\
&&&&-1&0\\
&&&&&\ddots\end{pmatrix}\in \cD^+_0,\text{ with }J^2 = -I;\end{equation}
and the associated Lie algebra order 2 involution
\begin{equation}\lab{2.27}
\cJ:\cD\to \cD:a\mapsto \cJ(a):=J a ^\top J.\end{equation}
The splitting into two Lie subalgebras\footnote{Note that $g_-$ is the fixed point set of $\cJ$.} (with corresponding projections $P_\pm$)
\begin{equation}\lab{2.28}
g=g_++g_-\hensp{with} g_+= \cD_- +\cD^-_0\hensp{and} g_-= \{a+\cJ a, \, a\in \cD\}=\mathrm{sp}(\infty),\end{equation}
with corresponding Lie groups\footnote{$G_+$ is the group of invertible elements in $g_+$, i.e., lower-triangular matrices, with nonzero $2\times 2$ blocks proportional to Id along the diagonal.} $G_+$ and $G_- =  \mathrm{Sp}(\infty)$, plays a crucial role here.  Notice that $g_- = \mathrm{sp}(\infty)$ and $G_+=\mathrm{Sp}(\infty)$ stand for the infinite rank affine symplectic algebra group.  Let (formally speaking)
$\lra{A,B} = \tr AB$ be the Ad-invariant inner product.  The applying \pref{13} of the \aks\ yields the flows \cite{23}:
\begin{equation}\lab{2.29}
\frac{\part L}{\part t_i}=[P_+ \nabla \cH_i,L] = [-P_- \nabla \cH_i,L],\enspace \cH_i = -\frac{\tr L^{i+1}}{i+1} \end{equation}
on matrices $L=Q\wedge Q^{-1}$, with $Q\in G_+$ and $\La$ the customary shift operator.  We call these equations the Pfaff Lattice.\footnote{The Hamiltonians $\cH_i$ are viewed as formal sums; the convergence of this formal sum would require some sufficiently fast decay of the entries of $L$.  Since $\nabla \cH_i=L^i$, one does not need to be concerned about this point.}  Note that $L$ of the above form are preserved by \pref{2.29}, as follows from Proposition \ref{p2.1}. 

Remembering the decomposition \pref{2.24},  write $a\in D$, $a=a_-+a_0+a_+$, while for any element $a\in g=g_++g_-$, write $a=P_+a+P_-a$, and then we have that
\begin{align*}
a&=a_-+a_0+a_+\\
&= P_+a +P_-a\\
&= \lrc{(a_- -\cJ a _+)+\frac12(a_0-\cJ a_0)}+\lrc{(a_+ +\cJ a_+)+\frac12(a_0+\cJ a_0)} ,\end{align*}
and so we can write the Pfaff Lattice \pref{2.29} more explicitly:
\begin{align}
\label{2.30}
\frac{\part L}{\part t_i}&= \lrb{-\lrp{(Li)_- - \cJ(L^i)_+} -\frac12\lrp{(L^i)_0-\cJ(L^i)_0},L}\\
&= \lrb{ \lrp{(L^i)_+ + \cJ(L^i)_+}+\frac12\lrp{(L^i)_0 +\cJ(L^i)_0},L}.\nn\end{align}
We have the followings proposition \cite{23}:
\begin{prop}\lab{p2.1}
For the matrices
\[
L:= Q\wedge Q^{-1}\hensp{and} m:= Q^{-1}J Q^{-1T},\hensp{with} Q\in G_+,\]
the following three statements are equivalent:
\begin{itemize}
\item[(i)] $\frac{\part Q}{\part t_i}Q^{-1}=-P_+L^i$,
\item[(ii)] $L^i+\frac{\part Q}{\part t_i} Q^{-1}\in g_-$,
\item[(iii)] $\frac{\part m}{\part t_i}  = \wedge^i m +m\wedge^{T^i}$
\end{itemize}
\end{prop}
which yields the following:
\begin{thm}\lab{t2.2}
Consider the skew-symmetric solution
\[
m_\infty(t) =e^{\sum t_k\wedge^k}m_\infty(0)e^{\sum t_k\wedge^{Tk}},\]
to the commuting equations
\begin{equation}\lab{2.31}
\frac{\part m_\infty}{\part t_k}=\wedge^k m_\infty +m_\infty \wedge^{Tk},\end{equation}
with skew-symmetric initial condition $m(0)$ and its {\bf skew-Borel decomposition}
\begin{equation}\lab{2.32}
m_\infty = Q^{-1}JQ^{-1T},\hensp{with} Q\in G_+.\end{equation}
Then the matrix $Q$ evolves according to the equations
\begin{equation}\lab{2.33}
\frac{\part Q}{\part t_i} Q^{-1} = -P_+ (Q\wedge^i Q^{-1}),\end{equation}
and the matrix $L:=Q\wedge Q^{-1}$ provides a solution to the Lax pair
\begin{equation}\lab{2.34}
\frac{\part L}{\part t_i}= [-P_+ L^i,L] = [P_- L^i,L].\end{equation}
Conversely, if $Q\in G_+$ satisfies \pref{2.33}, then $m_\infty$, defined by \pref{2.32}, satisfies \pref{2.31}.
\hfill\qed
\end{thm}

 \section{A.C.I. Examples of the AKS Theorem}
 This section deals with integrable systems \cite{1} that linearize on algebraic tori and, in particular, are solved by Lax-equations involving a formal parameter, i.e., Lax-equations on Kac-Moody Lie algebras.  This gives rise to an invariant algebraic curve (constant as time evolves) such that the systems are linearized on the Jabobeans of these curves.  The periodic matrices defining the Jacobeans ultimately yield the fundamental physical quasi-periods of these integrable systems.  Effective techniques for solving these equations are given .
 
 Set \[
 g:=\Big\{ \sum_{-\infty<i\ll \infty} a_i h^i\mid a_i+\gl(n,\R \hensp{or} \C)\Big\}=g_++g_-,\]
 with
 \begin{align*}
 g_-&:= \Big\{ \sum_{0\leq i \ll \infty} a_i h^i\in g\Big\} = g^\bot_- \cong g^\ast_+,\\
 g_+&:= \Big\{ \sum_{-\infty<i\leq -1}a_i h^i \in g\Big\} = \gbp \cong g^\ast_-.\end{align*}
 Here $h$ is a formal parameter, with the Ad-invariant form on $g$
 \[
 \lra{a,b}:= \hbox{coef}_{h^0}(\tr(abh)) = \sum_{i+j+1=0} (a_i,b_j),\]
 with $(\ensp,\ensp)$ being the Killing form on $\gl(n)$, and
 \[
 \Big[ \sum_i a_i h^i,\sum_j b,h^j\Big] = \sum_{i,j} [a_i,b_j] h^{i+j}.\]
 Thus the Hamiltonian structure on $\gbm\cong g^\ast_+$ given by \pref{7} is
 \[
 \cX_H(X)=P_- [\nabla_+ H(X),X],\]
 with the manifold
 \[
 C_m (\alpha,\gamma):=\alpha h^m +\gamma h^{m-1}+A_{m-1}\]
 invariant under the Hamiltonian action, so $\cX_H$ restricts to $C_m(\alpha,\ga)$, with (Casimirs) $\alpha$ and $\ga$ two fixed diagonal matrices and 
 \[
 A_{m-1}:= \Big\{ \sum^{m-1}_{j=0}a_j h^j \mid \diag (a_{j-1})=0\Big\},\]
 with
$
 \alpha=\diag(\alpha_1,\alpha_2,\dots,\alpha_n)$ with $\prod_{i<j}(\alpha_i-\alpha_j)\ne 0$.  Taking Hamiltonians of the form
 \begin{equation}\lab{3.1} 
 H(a) = \lra{f(ah^{-j}),h^k},\qquad a\in C_m(\alpha,\ga),\end{equation}
 with $f$ ``nice" and applying \pref{14} of the \aks\ we find \cite{11}
 \begin{equation}\label{3.2} 
 \dot a = \cX_H(a) = [a,(f'(ah^{-j})h^{k-j})_-],\qquad a\in C_m(\alpha,\ga),\end{equation}
 and upon setting $j=m$, $k=m+1$, this yields
 \begin{align} \label{3.3} 
 \dot a&= [a,b+\beta h],\; \beta:= f'(\alpha),\; b:=\ad_\beta \ad^{-1}_\alpha a_{m-1}+f''(\alpha)\ga,\\
 b_{ij}&= (1-\delta_{ij})(\beta_i-\beta_j) (\alpha_i-\alpha_j)^{-1}(a_{m-1})_{ij}+\delta_{ij}\ga_{ii} f''(\alpha_i) ,\notag\end{align}
 \begin{equation}   H(a) =\lra{f(ah^{-m}),h^{m+1}}.\lab{3.4} 
\end{equation}
 
 If in the above we set $m=1$, $\ga =0$, $\beta=\alpha^{1/2}$, we find the Euler-Arnold spinning top for the Lie algebra $\gl(n)$, while if we set $A_0 = -A ^\top_0$, we arrive at the Euler top for $u(n)$.  Moreover let us set for $x,y\in \R^n$ or $\C^n$
 \begin{equation}\lab{3.5} 
 \Ga_{xy} = x\otimes y - y\otimes x,\; \Ga_{xx}=x\otimes x, \; \Ga_{yy}=y\otimes y,\; \Delta_{xy} = x\otimes y+y\otimes x,\end{equation}
 and let us call the following differential equations,
 \begin{equation}\lab{3.6} 
 \dot\Ga_{xy} = [\Ga_{xy},\Ga],\qquad \Ga_{ij}=(\Ga_{xy})_{ij} (J_i+J_j)^{-1},\end{equation}
 for $J_i>0$, the ``special" Euler equations.  In what follows $J_i=\sqrt{\alpha_i}$, and we arrive at the following theorem of \cite{11}:
 
 \begin{thm}\lab{t3.1}
 The Lax equation $\dot a=[a,\Ga+\beta h]$, with $a=a(h)$, $\beta =\diag (\beta_1,\dots,\beta_n)$, $\Ga=\ad \beta \ad^{-1}_\alpha \Ga_{xy}$ of  {\rm (3.3)} corresponds to {\rm (a)} Euler equations, {\rm (b)} the geodesic flow on the ellipsoid and the Neumann problem, and {\rm (c)} the central force problem on the ellipsoid respectively, with
 \begin{align}\label{3.7} 
 {\rm (a)} &\qquad  a = \alpha h +\Ga_{xy}\\
 \nn {\rm (b)}&\qquad a= \alpha h^2 +h \Ga_{xy} - \Ga_{xx} \\
 \nn {\rm (c)}& \qquad a= \alpha h^2 +h\Ga_{xy}+\Delta_{xy} -\alpha\end{align}
 and with the respective Hamiltonians  \pref{3.4} of the  form
 \begin{alignat}{5}\label{3.8} 
 {\rm (a)} &\qquad  H =\lra{\frac23  (ah^{-1})^{3/2},h^2},&&\qquad f(x)=\frac{2}{3} x^{3/2}&\\ \nn
 {\rm (b)}&\qquad H=\lra{ \ln(ah^{-2}),h^3},&&\qquad f(x)=\ln(x)\hbox{ {\rm (geodesic)}}&\\ \nn
 &\qquad H= \lra{\frac12(ah^{-2})^2),h^3},&&\qquad f(x)=\frac12 x^2\hbox{ {\rm (Neumann)}}&\\ \nn
 {\rm (c)}&\qquad H=\lra{\ln (ah^{-2}),h^3},&&\qquad f(x)=\ln(x).\end{alignat}
 \end{thm}
 
 In order to study the equations \pref{3.2} we could apply \pref{16} of the \aks, thinking of $h$ as a complex variable and use the Birkhoff Factorization Theorem in the style of \cite{24, 25}, but instead we should take advantage of the additional algebraic geometrical structure in  \pref{3.2}.  This leads \cite{12, 1} us to the following propositions, theorems and corollaries, which enable us to solve flows of the form \pref{3.2} as linear flows on algebraic tori which are Jacobians of curves.  
 
 \begin{prop}\lab{p3.2}
 Given a Lax pair defining the flow
 \begin{equation}\lab{3.9} 
 \dot X(h)=[X(h),Y(h)],\qquad X(h),Y(h)\in \gl(N)[h,h^{-1}],\end{equation}
 the functions $q_{k\ell}$ that are defined by the coefficients of the characteristic polynomial of $X(h)$,
 \[
 \det(z \Id_N-X(h))=z^N +\sum_{\begin{subarray}{l} a\leq 1<N\\-\ell'\leq k\leq \ell \end{subarray}}
 q_{k\ell} h^k z^\ell \]
 are constants of motion of the flow.  The plane algebraic curve, associated to each $X(h)$, 
 \[
 \Ga_X:= \{ (h,z)\in \C\times \C\mid \det (z\Id_N-X(h))=0\},\]
 is preserved by the flow.  Similarly, for each $X(h)$ the ``isospectral" variety of matrices $A_c\subset M$ defined by 
 \[\displaylines{
 A_c:= \{X'(h)\mid X(h)\hensp{and} X'(h)\hbox{ have the same characteristic polynomial}\\
 \hfill\cr\hfill\hbox{ with all } q_{k\ell}=c_{k\ell}\}}\]
 is preserved by the flow.  For $X\in A_c$ such that $\Ga_X=: \Ga_c$ is smooth, let us denote its smooth compactification by $\ol{\Ga_c}$ and let  \[
 \{p_1,\dots,p_s\}:=\ol{\Ga_c}\bsl \Ga_c\]
 denote the points at infinity.  At each of these points $h$ has a zero or a pole, i.e., (possibly after relabeling) we have that
 \[
 \mathrm{ord}_{p_i}(h)=\bcs -\mu_i& 1\leq i\leq s'\\
 \mu_i&s'+1\leq i\leq s\ecs\]
 where $\mu_i>0$ for $i=1,\dots,s$.
 
 \end{prop}
 
 Assume generically in $c$, that $\Ga_c$ is nonsingular, and for a generic point  $(h,z)$ on $X(h)\in A_c$, the eigenspace  $\xi=\xi(z,X(h))$ of $X(h)$ with eigenvalue  $z$ is one dimensional.  By Cramer's rule, $\xi=(\xi_i)_{1\leq i\leq N}$, normalized such that $\xi_1 = 1$, is a meromorphic function on $\ol{\Ga_c}$.
 
 For a generic $X(h)\in A_c$, with corresponding normalized eigenvector $\xi$, let $\cD_X$ be the minimal effective divisor on 
 $\Ga_c$ such that
 \[
 (\xi_\ell)_{\Ga_c}\geq -\cD_X\quad \hensp{for all} \ell=1,\dots,N;\]
 by continuity, $d:= \deg (\cD_X)$ is independent of $X=X(h)\in A_c$ and thus, $\cD_X$ defines an effective divisor of degree $d$ in $\ol{\Ga_c}$ for any $X=X(h)\in A_c$.  The point is to study the motion of the divisor $\cD_X$ in $\Ga_c$, when $X(h)$ is moving in $A_c$.  Roughly speaking, $\cD_X$ is the divisor of poles of the normalized eigenvector $\xi(z; X(h))$ on $\Ga_c$, not at $\infty$.  Note for non-generic $X(h)$ the divisor $\cD_X$ may contain one or several of the points $p_i$ at infinity.
 
 Choose a divisor $\cD_0\in \Div^d(\ol{\Ga_c})$ and a basis $(\om_1,\dots,\om_g)$ of holomorphic differentials on $\ol{\Ga_c}$ and let $\vec\om := (\om_1,\dots,\om_g) ^\top$.  Define the map 
 \begin{align*}
 J_c: A_c&\to \mathrm{Jac}(\ol{\Ga_c})\\
 X&\mapsto \int^{\cD_X}_{\cD_0}\vec\om.\end{align*}
 For example, one may choose a base point $q$ on $\olgc$ and take $\cD_0:= dq$.  Then the map is given by
 \[
 J_c(X) = \sum^d_{i=1}\int^{q_i}_q\vec\om\in \Jac(\olgc),\]
 where $\cD_X=q_1+\cdots+q_d$.
 
 It is easy to check from \pref{3.9} that $\xi(t)=\xi(z,X(h,t))$ satisfies and defines a function $\la$ as follows:
 \begin{equation}\lab{3.10} 
 \dot\xi+Y\xi =: \la\xi,\end{equation}
 $Y=Y(h,X(h,t))$, with $\la$ a scalar function of $(h,z,t)$.  This leads to the following theorems.
 
 \begin{thm}\lab{t3.3}
 Along the integral curves $X(t)$  of the Lax equation $\dot X=[X,Y]$, the derivative of the linearizing map is given by
 \[
 \frac{d}{dt}\int^{\cD_{X(t)}}_{\cD_{X(0)}} \vec \om = \sum^s_{i=1}\Res_{p_i}\la(h,z,t)\vec\om.\]\end{thm}
 
 \begin{thm}[Linearization Criterion]\lab{t3.4}
 The map $J_c$ linearizes the spectral flow $\dot X=[X,Y]$ on $A_c$, that is to say
 \[
\int^{\cD_{X(t)}}_{\cD_{X(0)}}  \vec\om=t\sum^s_{i=1} \Res_{p_i}\la(h,z,X(h,0))\vec\om,\]
if and only if there exists for each $X\in A_c$ a meromorphic function $\phi_X$ on $\ol{\Ga_c}$ with $(\phi_X)_{\olgc}\geq -n\sum^{s'}_{i=1}\mu_i p_i+n'\sum^s_{i=s'+1}\mu_i p_i$ such that for all $p_i$,
\[
(\hbox{Laurent tail of $\frac{d\la(h,z,X)}{dt}$ at $p_i$})=(\hbox{Laurent tail of $\phi_X$ at $p_i$}),\]
where $(\Delta_{k,\ell}(z,X(h,t))$ being the $(k,\ell)$ cofactor of $z\Id_N-X(h,t))$
\[
\frac{d\la(h,z,X)}{dt} = \frac{d}{dt}\Big( \sum^N_{\ell=1} Y_{1\ell}(h,z,X(h,t)) \frac{\Delta_{1\ell} (z,X(h,t))}{\Delta_{11}(z,X(h,t))}\Big),\]
$d/dt$ being computed using the Lax equation $\dot X=[X,Y]$.
\end{thm}
 
 \begin{cor}\lab{c3.5}
 Suppose that $h$ has no zero at infinity and that there exists a polynomial $p(x,y,z)$ whose coefficients are arbitrary constants of the motion, and that there exists an algebraic function $\Psi$, whose coefficients are arbitrary constants of the motion, such that
 \[
 Y(h) = \Psi(p(X,h,h^{-1}))+ (C_0+C_1 h^{-1} + C_2 h^{-2}+\cdots)\]
 where $C_0$ is a lower triangular matrix, and where the matrices $C_1,C_2,\dots$ are arbitrary.  If $\xi/h$ has no pole at the points $p_i$, then the linearization criterion is satisfied.
 \end{cor}
 
 Theorem \ref{t3.4} (due to P. Griffiths \cite{PG}) can be used to show that the flows in \pref{3.2} and hence (3.3) and Theorem \ref{t3.1} linearize on $\Jac(\Ga_c)$, while Corollary \ref{c3.5} is more special but actually applies to many other cases found in \cite{4, 9, 5}, like the general periodic Toda flows.  This brings us to the following definitions, where $\{\ensp,\ensp\}$ is a Poisson bracket on the manifold $M$ and the $F_i$ are integrals in involution.
 
 \begin{defn}\lab{d3.6}
 Let $(M,\{\cdot,\cdot\},F)$ be a complex integrable system, where $M$ is a non-singular affine variety and where $F=(F_1,\dots,F_s)$.  We say that $(M,\{\cdot,\cdot\},F)$ is an \emph{algebraic completely integrable system} or an \emph{a.c.i. system} if for generaic $c\in \C^s$ the fiber $F_c$ defined by $F=c$ is an affine part of an Abelian variety and if the Hamiltonian vector fields $\cX_{F_i}$ are translation invariant, when restricted to these fibers. In the particular case in which $M$ is an affine space $C^n$ we will call $(C^m,\{\cdot,\cdot\},F)$ a \emph{polynomial a.c.i.\ system}.  When the generic Abelian variety of the a.c.i. system is irreducible we speak of an \emph{irreducible a.c.i.\ system.}
 \end{defn}
 
 \begin{defn}\lab{d3.7}
 Let $(M,\{\cdot,\cdot\},F)$ be a complex integrable system, where $M$ is a (non-singular) affine variety, and where $F=(F_1,\dots,F_s)$.  We say that $(M,\{\cdot,\cdot\},F)$ is a \emph{generalized a.c.i. system} if for generic $c\in \C^s$ the integrable vector fields $\cX_{F_1},\dots,\cX_{F_s}$ define the local action of an algebraic group on $F_c$.\end{defn}
 
 Theorems \ref{t3.4} and Corollary \ref{c3.5} are general tools to show systems like \pref{3.2} and (3.3) are a.c.i.\ systems \emph{in the appropriate coordinates} (those of $M$), which in fact are meromorphic functions on the Abelian variety $\Jac(\Ga_c)$.  For instance, in the case of Theorem \ref{t3.1}, the appropriate coordinates are $x^2_i,y^2_i$ and $x_iy_i$, $1\leq i\leq n$.  The Toda hierarchy \pref{2.12} is a generalized a.c.i.\ system, with the appropriate coordinates being the $a_i$ and $b_i$.
 
 \section{Random matrices, Limiting Distributions and $\mathrm{KdV}$}
 The point of this section is to get partial differential equations for the spectral gap probabilities for the Airy and Bessel processes, which are ``universal" limiting processes in the GUE ensemble respectively at the hard and soft edges.  In particular, for the one-interval gap case we recover respectively the Painl\'ev\`e II and V equations going with the Tracy-Widom distribution and the Bessel distribution.  Since these distribuations are ``universal,'' they appear in many other contexts in statistical mechanics.  The integrable deformation class of these universal distributions is seen to be the KdV equation and hence its vertex operator and Virasoro symmetries as well as Sato's KP theory  play  a crucial role in the derivation of the partial differential equations.

 Define on  the ensemble $\cH_N=\{N\times N$ Hermitian matrices$\}$ the probability
 \[
 P(M\in dM)=ce^{-\Tr V(M)}dM,\]
 where $c$ is a normalization constant.  Then for $z_1,\dots,z_N\in \R$ we have
 \[\displaylines{\qquad
 P\hbox{(one eigenvalue in each } [z_i,z_i+dz_i],\; i=1,\dots,N)\hfill\cr\hfill
 = c\vol (U(N)) e ^{-\sum^N_1 V(z_i)} \Delta^2 (z)dz_1\cdots dz_N\qquad}\]
 and for $0\leq k\leq N$,
 \begin{align*}
 &P\hbox{(one eigenvalue in each } [z_i,z_i+dz_i],\; i=1,\dots,k)\\
 &\qquad{} = c' \lrp{\int_{\R^{n-k}} \hspace*{-12pt}\cdots \int e^{-\sum^N_1 V(z_i)}\Delta^2(z)dz_{k+1}\cdots dz_N} dz_1\cdots dz_k\\
 &\qquad{} \stackrel{\ast}{=} c'' \det \lrp{K_N(z_i,z_j)}_{1\leq i,j\leq k} dz_1\cdots dz_k,\end{align*}
 and if $J\subset \R$, then
 \[
 P\hbox{(exactly $k$ eigenvalues in }J)=\frac{(-1)^k}{k!} \frac{\part^k}{\part \la^k}\det(I-\la K^J_N)\big|_{\la=1},\]
 where
 \[
 K^J_N(z,z') = K_N(z,z') I_J(z'),\]
 $I_J$ the indicator function of $J$, and $K_N$ is the Schwartz kernel of the orthogonal projector $\C[z]\to \C +\C z+\cdots + \C z^{N-1}$ with respect to the measure $e^{-\Tr V(z)}dz$, namely
 \[
 K_N(z,z')=\sum^{N-1}_{k=0} \vp_k (z)\vp_k(z')\]
 in terms of orthonormal functions $\vp_k(z)=e^{-\Tr V(z)/2}p_k(z)$ with respect to $dz$ or orthogonal polynomials $p_k(z)=(1/\sqrt{h_k}) z^k+\cdots $ with respect to $e^{-Tr V(z)}dz$.
 
 When $V(z)$ is quadratic, and more generally convex, we have for large $N$ \cite{26}
 \begin{align*}
 P\hbox{(an eigenvalue } \in [z,z+dz]) &= K_N(z,z)dz\\
 &\sim \bcs  \frac{1}{\pi}  (2N-z^2)^{1/2}dz&\hbox{if }|z|<(2N)^{1/2}\\
 0&\hbox{if } |z|>(2N)^{1/2}\ecs\end{align*}
 is given by the circular distribution (Wigner's semi-circe law). We have that for $z\sim 0$ the average spacing between the eigenvalues near the origin is $\sim (K_N(0,0))^{-1}=\pi/\sqrt{2N}$ and near the edge $(z\sim \sqrt{2N})$ is 
 $1/(2^{1/2}N^{1/6})$, leading to 
 \begin{align*}
 &\lim_{N\uparrow \infty} \frac{1}{K_N(0,0)} K_N\lrp{\frac{z}{K_N(0,0)},\frac{z'}{K_N(0,0)}}\\&\qquad{} =K(z,z')=
 \frac{1}{\pi} \frac{\sin \pi(z-z')}{z-z'}\quad\hbox{(bulk scaling limit),}\\
& \lim_{N\uparrow\infty} \frac{1}{K_N(0,0)} K_N  \lrp{\sqrt{2N} +\frac{z}{2^{1/2} N^{1/6}},
 \sqrt{2N}+\frac{z'}{2^{1/2}N^{1/6}} }\\&\qquad{}= \int^\infty_0 A(x+z)A(x+z')dx\quad\hbox{(edge scaling limit),}\end{align*}
 in terms of the classical Airy function.  In a similar context one also finds the Bessel kernel; for background on such matters, consult \cite{8}.
 
 This section deals with computing PDEs for the gap probabilities given by Fredholm determinants involving the limiting ``universal" Airy and Bessel kernel, which appear now in a variety of circumstances.  Miraculously these gap probabilities are given by ``continuous soliton formulas" of the KP equation and in particular the KdV equation, i.e., the 2-Gel$'$fand-Dickey equations.  This plus Virasoro constraints built into the very special KdV solutions associated with these kernels yield the PDEs.
 We shall be dealing with the KP hierarchy, briefly explained in Section 2; remember it is a hierarchy of isospectral deformations of a pseudodifferential operator $L=D+\sum_{i\geq 1} a_i(x,t)D^{-i}$, with $D:=d/dx$,
 \[
 \frac{\part L}{\part t_n}=[(L^n)_+,L]\quad\hbox{ for }t\in \C^\infty.\]
 We also consider the $p$-\gd\ hierarchy, i.e., the reduction to $L$'s such that $L^p$ is a differential operator for some fixed $p\geq 2$.  Sato tells us that the solution $L$ to  the KP equations   can ultimately be expressed in terms of a $\tau$-function.  The wave and adjoint wave functions, expressed in terms of the $\tau$-function \cite{27}
 \begin{equation}\lab{4.1} 
 \Psi(x,t,z)  = e^{xz+\sum^\infty_1 t_i z^i} \frac{\tau(t-[z^{-1}])}{\tau(t)}, \enspace
 \Psi^\ast (x,t,z) = e^{-xz-\sum^\infty_1 t_i z^i} \frac{\tau(t+[z^{-1}])}{\tau(t)} \end{equation}
 satisfy
 \begin{alignat}{5}\label{4.2} 
 z\Psi&=L\Psi &\quad z\Psi^\ast& = L^\top \Psi^\ast \\
 \frac{\part \Psi}{\part t_n}&=(L^n)+\Psi&\qquad \frac{\part \Psi^\ast}{\part t_n}&=-(L^{\top n})_+  \nonumber\Psi. \end{alignat}
 As in the general theory of integrable systems, vertex operators play a prominent role: they are Darboux transforms involving all times.  In particular, for the KP equation, the vertex operator
 \begin{equation} \lab{4.3}
 X(t,y,z) := \frac{1}{z-y} e^{\sum^\infty_1 (z^i-y^i) t_i } e^{\sum^\infty_1 (y^{-i}-z^{-i}) \frac{1}{i} \frac{\part}{\part t_i}} \end{equation}
 has the striking feature that $X(t,y,z)\tau$ and $\tau+X(t,y,z)\tau$ are both $\tau$-functions.  Given distinct roots of unity $\om,\om'\in \zeta_p := \{\om\mid\om^p = 1\}$, the vertex operator $X(t,\om z,\om' z)$ maps the space of $p$-Gel$'$fand-Dickey $\tau$-functions into itself.
 
 We also note that the $2$-\gd\ KP equation satisfied by $\tau$ is
 \begin{equation}\lab{4.4}
 \lrp{ \lrp{\frac{\part}{\part t_1}}^4 - 4\frac{\part^2}{\part t_i \part t_3}}\log \tau + 6 \lrp{\frac{\part^2}{\part t^2_1} \log\tau}^2 = 0.\end{equation}
 
 We have the two basic theorems \cite{13}:
 \begin{thm}\lab{t4.1}
 Define the $(x,t)$-dependent kernel $k_{x,t}(y,z)$ and $k^E_{x,t}(y,z):= k_{x,t}(y,z) I_E(z)$ with $x\in \R$, $t\in \C^\infty$, $y,z\in \C$, and $E\subset \R^+$ a Borel subset:
 \begin{equation} \lab{4.5}
 k_{x,t}(y,z):= \int^x dx \sum_{\om\in\zeta_p}a_\om \Psi^\ast (x,t,\om y) \sum_{\om'\in\zeta_p} b_{\om
'} \Psi(x,t,\om' z),\end{equation}
where $\Psi(x,t,z)$ and $\Psi^\ast(x,t,z)$ are the wave and adjoint wave function for the $p$-Gel\,{\cprime}fand-Dickey hierarchy and where the coefficients $a_\om,b_\om\in\C$ are subjected to $\sum_{\om\in \zeta_p} \frac{a_\om b_\om}{\om}=0$.  Then the following holds:
{\setbox0\hbox{(ii)}\leftmargini=\wd0 \advance\leftmargini\labelsep
 \begin{itemize}
\item[(i)] The kernel $k(y,z)$, its determinant and its Fredholm determinant are all three expressible in terms of the vertex operator
\begin{equation}\lab{4.6}
Y(t,y,z) := \sum_{\om,\om'\in \zeta_p} a_\om b_{\om'} X(t,\om y,\om' z)\end{equation}
action on the underlying $\tau$-function:
\begin{align} \label{4.7}
k_{x,t}(y,z)&= \frac{1}{\tau} Y(t,y,z)\tau\\
\det (k_{x,t} (y_i,z_j))_{1\leq i, j\leq n} &= \frac{1}{\tau} Y(t,y_1,z_1)\cdots Y(t,y_n,z_n)\tau\nn \\\nn
\det (I-\mu k^E_{x,t}) &=\frac{1}{\tau} e^{-\mu \int_E dz\; Y(t,z,z) }\tau \quad\hbox{(``continuous" soliton formula)}.\end{align}
\item[(ii)] Let the kernel $k_{x,t}(y,z)$ in \pref{4.5} be such that the underlying $\tau$-function of $\Psi$ and $\Psi^\ast$ satisfies a Virasoro constraint:\footnote{Define $W^{(0)}_n=\delta_{n,0}$,
\[
J^{(1)}_n:= W^{(1)}_n = \bcs \part/\part t_n&\hbox{if }n>0\\
(-n)t_{-n}&\hbox{if }n<0\\
0&\hbox{if }n=0\ecs,\quad J^{(2)}_n :=W^{(2)}_n+(n+1)W^{(1)}_n=\sum_{i+j=n} : J^{(1)}_i J^{(1)}_j:\]}
\[W^{(2)}_{kp}\tau = c_{kp}\tau\hbox{ for a fixed }k\geq -1.\]
Then for the disjoint union $E=\bigcup^r_{i=1}[a_{2i-1},a_{2i}]\subset \R_+$, the Fredholm determinant $\det(I-\mu k^E_{x,t})$ satisfies the following Virasoro constraint for that same $k\geq -1$:
\begin{equation}\lab{4.8}
\lrp{-\sum^{2r}_{i=1} a^{kp+1}_i \frac{\part}{\part a_i} +\frac12 (W^{(2)}_{kp} -c_{kp} )} \lrp{\tau\det(I-\mu k^E_{x,t})} = 0;\end{equation}
note the boundary $a$-part and the time $t$-part decouple.

\end{itemize}
} \end{thm}

In the next theorem we apply equation \pref{4.8} to compute the partial differential equations for the distribution of the spectrum for matrix ensembles whose probability is given by a kernel.  To state the problem, consider a first-order differential operator $A$ in $z$ of the form
\begin{equation}\lab{4.9}
A=A_z = \frac12 z^{-m+1}\lrp{\frac{\part}{\part z}+V'(z)} + \sum_{i\geq 1}c_{-2i} z^{-2i},\end{equation}
with
\begin{equation}\lab{4.10}
V(z)=\frac{\alpha}{2} z+\frac{\beta}{6} z^3\not\equiv 0,\quad m=\deg V' = 0\hensp{or} 2,\end{equation}
and the differential part of its ``Fourier" transform
\begin{equation}\lab{4.11}
\hat A=\hat A_x =\lrp{\frac12 (x+V'(D)) D^{-m+1}+\sum_{i\geq 1} c_{-2i} D^{-2i}}_+\hensp{with} D=\frac{\part}{\part x}.\end{equation}
Given a disjoint union $E=\bigcup^r_{i=1} [A_{2i-1}, A_{2i}]\subset \R^+$, define differential operators $\cA_n$, which we declare to be of homogeneous ``weight" $n$, as follows:
\[\cA_n:= \sum^{2r}_{i=1} A_i^{\frac{n+1-m}{2}} \frac{\part}{\part A_i},\qquad n=1,3,5,\dots\;.\]
We can now state the second main theorem.

\begin{thm}\lab{t4.2}
Let $\Psi(x,z)$, $x\in\R$, $z\in \C$ be a solution of the linear partial differential equation
\begin{equation}\lab{4.12}
A_z\Psi(x,z) = \hat A_x \Psi(x,z),\end{equation}
with holomorphic (in $z^{-1}$) initial condition at $x=0$, subjected to the following differential equation for some $a,b,c\in \C$,
\begin{equation}\lab{4.13}
(a A^2_z +b A_z +c)\Psi(0,z) = z^2\Psi(0,z),\hensp{with} \Psi(0,z)=1+O(z^{-1}).\end{equation}
Then
{\setbox0\hbox{()}\leftmargini=\wd0 \advance\leftmargini\labelsep
 \begin{itemize}
\item $\Psi(x,z)$ is a solution of a second order problem for some potential $q(x)$
\begin{equation}\lab{4.14}
(D^2+q(x))\Psi(x,z) = z^2\Psi(x,z).\end{equation}
\item Given the kernel
\begin{equation}\lab{4.15}
K^E_x (y,z):= I_E(z)\int^x \frac{\Phi(x,\sqrt{y})\Phi(x,\sqrt{z})}{2y^{1/4} z^{1/4}} dx,\end{equation}
with
\[
\Phi(x,u) := \sum_{\om = \pm 1} b_\om e^{\om V(u)} \Psi(x,\om u),\]
the Fredholm determinant $f(A_1,\dots,A_{2r}):= \det (I-\la K^E_x)$ satisfies a hierarchy of bilinear partial differential equations\footnote{The $p_i$ are the elementary Schur polynomials $e^{\sum^\infty_1 t_i z^i} = \sum^\infty_0 p_n(t)z^n$, and $p_i(\pm \tilde\cA):= p_i(\pm \cA_1,0,\pm \frac13 \cA_3,0,\dots)$.} in the $A_i$ for odd $n\geq 3$:
\begin{align}
\label{4.16}
&f\cdot \cA_n \cA_1 f-\cA_n f\cdot \cA_1 f- \sum_{i+j=n+1} p_i (\tilde \cA) f\cdot p_j (-\tilde\cA)f\\
\nn &\qquad\quad{}+ \hbox{(terms of lower weight $i$ for } 1\leq i\leq n)=0,\end{align}
where $x$ appears in the coefficients of lower weight terms only.\end{itemize}}   \end{thm}

As an application of Theorem \ref{t4.2} we find \cite{13}
\begin{thm}\lab{t4.3}
Given the Airy kernel
\[
K^E_x(y,z)=I_E(z)\frac{1}{2\pi} \int^x A(x+y)A(x+z)dx,\]
the Fredholm determinant $f(A_1,\dots,A_{2r}):= \det(I-\la K^E_x)$ satisfies the hierarchy of bilinear partial differential equations in the $A_i$ for odd $n\geq 3$:
\begin{equation}\lab{4.17}
f\cdot \cA_n \cA_1 f-\cA_n f\cdot \cA_1f-\sum_{i+j=n+1} p_i (\tilde\cA) f\cdot p_j(-\tilde\cA) =0,\end{equation}
with
\begin{equation}\lab{4.18}
\cA_n = \sum^{2r}_{i=1} A^{\frac{n-1}{2}}_i \frac{\part}{\part A_i}, \qquad n=1,3,5,\dots\;.\end{equation}
The variables appearing in the Schur polynomials $p_i$ are non-commutative and are written according to a definite order.  Finally, the first equation in the hierarchy {\rm (4.16)} takes on the following form:
\begin{equation}\lab{4.19}
\lrp{\cA^3_1 - 4\lrp{\cA_3-\frac12}} R+6(\cA_1 R)^2=0\end{equation}
for
\[
R:= \cA_1 \log f = \sum^{2r}_1 \frac{\part\log \det(I-\la K^E)}{\part A_i}.\]
When $E=(-\infty,A)$, the function $R=\cA_1\log f=\frac{\part}{\part A} \log\det(I-K^E)$ satisfies
\[
R^{'''} - 4A R' +2R + 6R^{'2}=0 \quad \hbox{{\rm (Painlev\'e II)}}.\]\end{thm}

\begin{proofn}
The Painlev\'e II equation for the logarithmic derivative has been obtained previously by Tracy and Widom \cite{TW}; equation \pref{4.19}, which leads to Painlev\'e, is new.  Setting in \pref{4.10} $V(z)=\frac{2}{3} z^3$, from \pref{4.9} and \pref{4.11} find
\begin{equation}\lab{4.20}
A: = \frac{1}{2z} \lrp{\frac{\part }{\part z}+2z^2} -\frac14 z^{-2} \hensp{and} \hat A=\frac{\part}{\part x}.\end{equation}
Then in terms of the Airy function
\[
F(u) := \int^\infty_{-\infty} e^{-\frac{y^3}{3}+ y u} dy,\]
$\Psi(x,z)$ has the following expression:
\[
\Psi(x,z)=\frac{1}{\sqrt{\pi}} e^{-\frac{2}{3} z^3} \sqrt{z} F(x+z^2)\]
and is a solution of
\[
A\Psi (x,z) = \hat A\Psi(x,z),\]
with $\Psi(0,z)$ satisfying
\[
\hat A^2 \Psi(0,z)=z^2 \Psi(0,z)\hensp{and} \Psi (0,z)=1+O(z^{-1}).\]
Setting in \pref{4.5} $b_+=1$ and $b_-=0$, we find for $\Phi(x,u)$ and $K^E_x(y,z)$ in \pref{4.15},
\[
\Phi(x,u) = \frac{\sqrt{u}}{\sqrt\pi}  A(x+u^2).\]
Then from Theorem \ref{t4.2}, $f(A_1,\dots,A_{2r}):= \det(I-\la K^E_x)$ satisfies the hierarchy of equations (4.16), with lower weight terms; the $\cA_n$ are defined by \pref{4.18}.  However, upon rewriting the variables of $p_n$ in an appropriate way, all lower weight terms can be removed, as follows from a combinatorial argument.

Finally from \pref{4.8} with $p=2$, $kp\to k$, $a^{kp}_i\to A^k_i$, ``$E\to E^2$,"  $t_3\to t_3+\frac{2}{3}$, $\tau^E:= \tau\det (I-\la K^E_{x,t})$ satisfies 
\begin{align*}
\cA_1 \log\tau^E&=\frac12 \lrp{\sum_{i\geq 3} it_i \frac{\part}{\part t_{i-2}} + 2\frac{\part}{\part t_1}} \log \tau^E + \frac{t^2_1}{4}\\
\cA_3\log\tau^E&= \frac12\lrp{\sum_{i\geq 1} it_i \frac{\part}{\part t_i}+2\frac{\part}{\part t_3}} \log\tau^E+ \frac{1}{16}\end{align*}
from which the partial derivatives $\frac{\part}{\part t_i}\log\tau^J$, $\frac{\part}{\part t_3}\log \tau^J$ and $\frac{\part^2}{\part t_1 \part t_3}\log \tau^J$ at $t= 0$ can be extracted.  Putting these partials in the KP-equation \pref{4.4} leads to the equation \pref{4.19}.  In the particular case of a semi-infinite interval $(-\infty,A)$, one finds the Painlev\'e II equation.\end{proofn}

We also have another application \cite{13} of Theorem \ref{t4.2}:
\begin{thm}\lab{t4.4}
Given the (Bessel) kernel
\begin{equation}\lab{4.21}
K^E (y,z) = -\frac{1}{2} I_E(z)\int^1_0 sJ_\nu (s\sqrt{y})J_\nu (s\sqrt{z})ds,\end{equation}
the Fredholm determinant $f(A_1,\dots,A_{2r}):= \det(I-\la K^E_x)$ satisfies the hierarchy  {\rm (4.16)} of bilinear partial differential equations in the $A_i$ for odd $n\geq 3$, with
\begin{equation}\lab{4.22}
\cA_n:= \sum^{2r}_{i=1} A^{\frac{n+1}{2}} \frac{\part}{\part A_i},\qquad n=1,3,5,\dots \;.\end{equation}
The first equation in the hierarchy {\rm (4.16)} for $F:=\log\det(I-\la K^E)$ takes on the following form:
\begin{align}\label{109}
 \lrp{\cA^4_1 - 2\cA^3_1 + (1-\nu^2)\cA^2_1 +\cA_3 \lrp{\cA_1-\frac12}} F\qquad&\\
 -4(\cA_1 F)(\cA^2_1 F) +6(\cA^2_1 F)^2=0;&\notag\end{align}
when $E=(0,A)$, we have for $R: =-\cA_1 F=-A\frac{\part}{\part A}\log\det (I-\la K^E)$, the equation \cite{TWb}
\[
A^2 R^{'''} +AR'' +(A-\nu^2)R' -\frac{R}{2} +4RR' -6AR^{'2} =0 \quad\hbox{{\rm (Painlev\'e V)}}.\]
\end{thm}
\begin{proofn}
Pick  in \pref{4.10} $V(z)=-z$; then from \pref{4.9} and \pref{4.11}
\[
A_z = \frac12 z \lrp{\frac{\part}{\part z}-1}\hensp{and} \hat A_x = \frac12(x-1)\frac{\part}{\part x}.\]
We look for a function $\Psi(x,z)$ satisfying
\begin{equation}\lab{4.23}
A_z\Psi(x,z) = \hat A_x \Psi(x,z)\end{equation}
with initial condition $\Psi(0,z)$ satisfying
\begin{equation}\lab{4.24}
\lrp{4A^2_z-2A_z - \nu^2 +\frac14} \Psi (0,z)=z^2\Psi(0,z),\quad \Psi(0,z)=1+ O\lrp{\frac1z} .\end{equation}
The solution to the differential equation \pref{4.24} is given by
\begin{align*}
\Psi(0,z) = B(z)&= \vareps \sqrt{z} H_\nu(iz)\\
&= \frac{e^z 2^{\nu+1/2}}{\Ga(-\nu+1/2)} \int^\infty_1 \frac{z^{-\nu+1/2}e^{-uz}}{(u^2-1)^{\nu+1/2}} du\\
&= 1+O\lrp{\frac1z}.\end{align*}
with $\vareps = i\sqrt{\pi/2} e^{i\pi\frac{\nu}{2}}$, $-\frac12<\nu<\frac12$.  Then
\[
\Psi(x,z) = e^{xz} B((1-x)z)\]
satisfies \pref{4.23}; from Theorem \ref{t4.2}, $\Psi(x,z)$ satisfies a second order spectral problem, which can explicitly be computed:
\[
\lrp{\frac{d^2}{dx^2}-\frac{(\nu^2-\frac14)}{(x-1)^2}} \Psi(x,z)=z^2\Psi (x,z).\]
Picking in \pref{4.5} $b_+ = e^{-i\pi \nu}/2\sqrt{\pi}$ and $b_-=i\bar b_+$, yield for \pref{4.15}
\begin{align*}
\Phi(x,z)&= \frac{e^{-i\pi \nu/2}}{2\sqrt{\pi}} e^{-z}\Psi(x,z)+\frac{ie^{i\pi\nu/2}}{2\sqrt{\pi}} e^z \Psi(x,-z)\\
&=\sqrt{\frac{(x-1)z}{2}} J_\nu ((1-x)iz),\end{align*}
and
\begin{align*}
K^E_x&= I_E(z) \int^x_1 \frac{\Phi(x,\sqrt{y})\Phi(x,\sqrt{z})}{2y^{1/4} z^{1/4}} dx\\
&= -\frac12 I_E(z)\int^{(1-x)i}_0 sJ_\nu(s\sqrt{y}) J_\nu (s\sqrt{z}) ds.\end{align*}
The special value $x=i+1$ leads to the standard Bessel kernel:
\begin{align*}
K^E_{1+i}&= -\frac12 I_E(z)\int^1_0 sJ_\nu(s\sqrt{y}) J_\nu (s\sqrt{z})ds\\
&=I_E(z) \frac{J_\nu(\sqrt y) \sqrt z J'_\nu(\sqrt z)-J_\nu (\sqrt z) \sqrt y J'_\nu (\sqrt y)}{
2(z-y)}.\end{align*}
From Theorem \ref{t4.2}, the Fredholm determinant 
\[
f(A_1,\dots,A_{2r}):= \det(I-\la K^E_{1+i})\]
satisfies equation (4.16), with
\[
\cA_n =\sum^{2r}_{i=1} A^{\frac{n+1}{2}}_i \frac{\part}{\part A_i},\qquad n=1,3,5,\dots\;.\]
Picking the special value $x=i+ 1$ leads to the shift $t_1\mapsto t_1+i+1$ and \pref{4.8} with $p=2$, $kp\to k$, $a^{kp}_i\to A^k_i$, etc.\ leads to
\begin{align*}
\cA_1\log \tau^E&= \frac12\lrp{\sum_{i\geq 1} it_i \frac{\part}{\part t_i} +\sqrt{-1} \frac{\part}{\part t_1}}
\log \tau^E +\frac14 \lrp{\frac14-\nu^2}\\
\cA_3 \log\tau^E&=\frac12\lrp{\sum_{i\geq 1}it_i \frac{\part}{\part t_{i+2}}+\frac12 \frac{\part^2}{\part t^2_1} + \sqrt{-1} \frac{\part}{\part t_3}} \log \tau^E\\
&\quad{} + \frac14\lrp{\frac{\part}{\part t_1}\log \tau^E}^2;\end{align*}
expressing the partial derivatives appearing in \pref{4.4} at $t=0$ in terms of the operators $\cA_1=\sum^{2r}_1 A_i \frac{\part}{\part A_i}$ and $\cA_3 = \sum^{2r}_1 A^2_i \frac{\part}{\part A_i}$ leads to the partial differential equation (4.23), which for $E=(0,A)$ leads to the Painlev\'e V equation, ending the proof of Theorem \ref{t4.4}.
\end{proofn}

\section{The distribution of the spectrum in Hermitian, symmetric and symplectic random ensembles  and their relation to the Toda and the Pfaff Lattices}

In this section we derive partial differential equations and partial differential-recursion relations for the spectral gap probabilities for the Gaussian and Laguerre ensemble for the Hermitian, symmetric and symplectic cases of random matrix ensemble \cite{14}.  The crucial tool is that the deformation classes for these ensembles are the (integrable) Toda and Pfaff lattices, and we rely heavily on the tau-function theory of these lattices as well as the Virasoro symmetries coming from the gauge transformations inherent in random matrix integrals.

Consider the weights of the form $\rho(z)dz:= e^{-V(z)} dz$ on an interval $F=[A,B]\subseteq \R$, with rational logarithmic derivative and subjected to the following boundary conditions:
\begin{equation}\lab{5.1}
-\frac{\rho'}{\rho}=V'=\frac{g}{f} = \frac{\sum^\infty_0 b_i z^i}{\sum^\infty_0 a_i z^i},\enspace
\lim_{z\to A,B} f(z)\rho(z) z^k = 0\hensp{for all} k\geq 0,\end{equation}
together with a disjoint union of intervals,
\begin{equation}\lab{5.2}
E=\bigcup^r_1 [c_{2i-1},c_{2i}]\subseteq F\subseteq \R.\end{equation}
The data \pref{5.1} and \pref{5.2} define an algebra of differential operators
\begin{equation}
\lab{5.3}
\cB_k  = \sum^{2r}_1 c^{k+1}_i f(c_i) \frac{\part}{\part c_i}.\end{equation}
Let $\cH_n$, $\cS_n$ and $\cT_n$ denote the Hermitian $(M=\ol M^\top)$, symmetric $(M=M^\top)$ and ``symplectic" ensembles $(M=\ol M^\top$, $M=J\ol M J^{-1})$, respectively.  Traditionally, the latter is called the ``symplectic ensemble," although the matrices involved are not symplectic!  These  conditions guarantee the reality of the spectrum of $M$.  Then, $\cH_n(E)$, $\cS_n(E)$ and $\cT_n(E)$ denote the subsets of $\cH_n$, $\cS_n$ and $\cT_n$ with spectrum in the subset $E\subseteq F\subseteq \R$.  The aim of this section is to find PDEs for the probabilities
\begin{equation} \label{5.4}
\begin{aligned}
P_n(E): & =P_n\hbox{ (all spectral points of $M\in E$)}\\
&= \frac{\int_{\cH_n(E), \cS_n(E)\text{ or }\cT_n(E)} e^{-\tr V(M)} dM}{
\int_{\cH_n(F),\cS_n(F)\text{ or }\cT_n(F) } e^{-\tr V(M)}dM}\\
&=\frac{ \int_{E^n} |\Delta_n(z)|^\beta \prod^n_{k=1}e^{-V(z_k)}dz_k}{
\int_{F_n}|\Delta_n(z)|^\beta \prod^n_{k=1} e^{-V(z_k)} dz_k},\quad \beta=2,1,4\text{ respectively }\end{aligned}
\end{equation}
for the Gaussian, Laguerre and Jacobi weights.

The method used in \cite{14} to obtain these PDEs involves inserting time-parameters into the integrals, appearing in \pref{5.4} and to notice that the integrals obtained satisfy
{\setbox0\hbox{()}\leftmargini=\wd0 \advance\leftmargini\labelsep
 \begin{itemize}
\item Virasoro constraints: linear PDEs in $t$ and the boundary points of~$E$, and
\item integrable hierarchies: satisfied by matrix integrals:
\[
\begin{array}{l|c|l}
\text{ensemble}&\beta&\text{lattice}\\
\hline
\text{Hermitian}&\beta=2&\text{Toda}\\
\text{symmetric}&\beta=1&\text{Pfaff}\\
\text{symplectic}&\beta=4&\text{Pfaff}\end{array}\]
\end{itemize}}  \noindent 
As a consequence of duality between $\beta$-Virasoro generators under the map $\beta\mapsto 4/\beta$, the PDEs obtained have a remarkable property: the coefficients $Q$ and $Q_i$ in the PDEs are functions of the variables $n,\beta,a,b$, and have the invariance property under the map
\[n\to -2n,\; a\to -\frac{a}{2},\; b\to -\frac{b}{2};\]
to be precise,
\begin{equation}\lab{5.5}
Q_i\lrp{-2n,\beta,-\frac{a}{2},-\frac{b}{2}}\Big|_{\beta=1} - Q_i(n,\beta,a,b)|_{\beta=4}.\end{equation}

\subsubsection*{Important remark} For $\beta=2$, the probabilities satisfy PDEs in the boundary points of $E$, whereas in the case $\beta=1,4$, the equations are inductive.  Namely, for $\beta=1$ (resp.\ $\beta=4$), the probabilities $P_{n+2}$ (resp.\ $P_{n+1}$) are given in terms of $P_{n-2}$ (resp.\ $P_{n-1}$) and a differential operator acting on $P_n$.

\subsection{Virasoro constraints}

\begin{thm}[Adler-van Moerbeke \cite{15}]\lab{t5.1}
The multiple integrals
\begin{equation}\lab{5.6}
I_n(t,c;\beta):= \int_{E^n} |\Delta_n(z)|^\beta \prod^n_{k=1} \lrp{e^{\sum^\infty_1 t_i z^i_k} \rho(z_k)dz_k}\hensp{for} n>0\end{equation}
and
\begin{equation}\lab{5.7}
I_n\lrp{t,c;\frac{4}{\beta}}:= \int_{E^n} |\Delta_n(z)|^{4/\beta} \prod^n_{k=1} \lrp{e^{\sum^\infty_1 t_i z^i_k} \rho(z_k) dz_k} \hensp{for} n>0,\end{equation}
with $I_0=1$, satisfying respectively the following Virasoro constraints\footnote{When $E$ equals the whole range $F$, then the $\cB_k$'s are absent in the formula \pref{5.8}.}
for all $k\geq -1$:
\begin{equation}\label{5.8}
\begin{aligned}
&\lrp{-\cB_k + \sum_{i\geq 0} \lrp{ a_i \, ^\beta \J^{(2)}_{k+i,n} (t,n) - b_i\, ^\beta \J^{(1)}_{k+i+1,n}(t,n)}} I_n(t,c;\beta)=0,\\
&\lrp{-\cB_k+\sum_{i\geq 0} \lrp{a_i \,^\beta \J^{(2)}_{k+i,n} \lrp{-\frac{\beta t}{2}, -\frac{2n}{\beta}}\right.\right.\\
&\hspace*{1in}\left.\left. +
\frac{\beta b_i}{2} \,^\beta \J^{(1)}_{k+i+1,n} \lrp{-\frac{\beta t}{2}, \frac{2n}{\beta}}}} I_n \lrp{t,c;\frac{4}{\beta}}=0,\end{aligned}\end{equation}
in terms of the coefficients $a_i$, $b_i$ of the rational function $(-\log \rho)'$ and the end points $c_i$ of the subset $E$, as in \pref{5.1}--\pref{5.3}.  For all $n\in \Z$, the ${}^\beta \J^{(2)}_{k,n}(t,n)$ and ${}^\beta \J^{(1)}_{k,n}(t,n)$ form a Virasoro and a Heisenberg algebra respectively, with central charge
\begin{equation}\lab{5.9}
c=1-6\lrp{\lrp{\frac{\beta}{2}}^{1/2} - \lrp{\frac{\beta}{2}}^{-1/2}}^2 .\end{equation}
\end{thm}

\begin{rem}
The $\bj^{(2)}_{k,n}$'s are defined as follows:
\begin{equation}\lab{5.10}
\bj^{(2)}_{k,n}=\frac{\beta}{2}\sum_{i+j=k} :\bj^{(1)}_{i,n} \bj^{(1)}_{j,n}:+ \lrp{1-\frac{\beta}{2}} \lrp{(k+1)\bj^{(1)}_{k,n} - k\J^{(0)}_{k,n}}.\end{equation}
Componentwise, we have
\[
\bj^{(1)}_{k,n} (t,n)=\bJ^{(1)}_k +nJ^{(0)}_k \hensp{and} \bj^{(0)}_{k,n} = nJ^{(0)}_k=n\delta_{0k}\]
and hence
\begin{align*}
\bj^{(2)}_{k,n}(t,n)&= \lrp{\frac{\beta}{2}} \bJ^{(2)}_k + \lrp{n\beta+(k+1)\lrp{1-\frac{\beta}{2}}}\bJ^{(1)}_k\\
&\quad{}+ n\lrp{(n-1) \frac{\beta}{2}+1} J^{(0)}_k.\end{align*}
\end{rem}

Setting
\[
dI_n(x):= |\Delta_n(x)|^\beta \prod^n_{k=1} \lrp{e^{\sum^\infty_1 t_i x^i_k} \rho(x_k)dx_k},\]
Theorem \ref{t5.1} is  based on the following  variational formula:
\[
\frac{d}{d\vareps} dI_n(x_i\mapsto x_i+\vareps f(x_i)x^{k+1}_i)\big|_{\vareps=0} = \sum^\infty_{\ell=0} \lrp{a_\ell \bJ^{(2)}_{k+\ell,n} - b_\ell \bj^{(1)}_{k+\ell+1,n}}dI_n.\]

\subsection{Matrix integrals and associated integrable systems}
\subsubsection*{Hermitian matrix integrals and the Toda lattice}  Given a weight $\rho(z)=e^{-V(z)}$ defined as in \pref{5.1}, the inner-product
\begin{equation}\lab{5.11}
\lra{f,g}_t = \int_E f(z) g(z) \rho_t (z)dz,\quad\hbox{ with }\rho_t := e^{\sum^\infty_1 t_iz^i} \rho(z) \end{equation}
leads to a moment matrix
\begin{equation}\lab{5.12}
m_n(t)=(\mu_{ij}(t))_{0\leq i,j<n} = \lrp{\lra{z^i,z^j}_t}_{0\leq i,j<n},\end{equation}
which is a H\"ankel matrix,\footnote{H\"ankel means $\mu_{ij}$ depends on $i+j$ only.}
thus symmetric.  H\"ankel is tantamount to $\La m_\infty = m_\infty \La^\top$.  The semi-infinite moment matrix $m_\infty$ evolves in $t$ according to the equations
\begin{equation}\lab{5.13}
\frac{\part \mu_{ij}}{\part t_k} = \mu_{i+k,j},\hensp{and thus} \frac{\part m_\infty}{\part t_k} = \La^k m_\infty \qquad\begin{pmatrix} \text{commuting}\\
\text{vector fields}\end{pmatrix}.\end{equation}
Another important ingredient is the factorization of $m_\infty$ into a lower- times an upper-triangular matrix\footnote{This factorization is possible for those $t$'s for which $\tau_n(t):= \det m_n(t)\ne 0$ for all $n>0$.}
 \[
m_\infty(t)=S(t)^{-1} S(t)^{\top-1},\]
where $S(t)$ is lower-triangular with nonzero diagonal elements. The following theorem can be found in \cite{28}. 

\begin{thm}\lab{t5.2}
The vector $\tau(t)=(\tau_n(t))_{n\geq 0}$, with
\begin{equation}\lab{5.14}
\tau_n(t):= \det m_n(t)=\frac{1}{n!} \int_{E^n} \Delta^2_n (z) \prod^n_{k=1} \rho_t (z_k) dz_k\end{equation}
satisfies
{\setbox0\hbox{(iv)}\leftmargini=\wd0 \advance\leftmargini\labelsep
 \begin{itemize}
\item[(i)] \emph{Virasoro constraints} \pref{5.8} for $\beta=2$,
\begin{equation}\lab{5.15}
\lrp{ - \sum^{2r}_1 c^{k+1}_i f(c_i) \frac{\part}{\part c_i}+\sum_{i\geq 0} \lrp{ a_i\J^{(2)}_{k+i}-b_i \J^{(1)}_{k+i+1}}}\tau=0.\end{equation}
\item[(ii)]The $\KP$-hiearchy\footnote{For the customary Hirota symbol $p(\part_t) f\circ g: = p\lrp{\frac{\part}{\part y}} f(t+y)g(t-y)\big|_{y=0}$, the $p_\ell$'s are the elementary Schur polynomials $e^{\sum^\infty_1 t_i z^i} := \sum_{i\geq 0} p_i(t_1,t_2,\dots)z^i$ and $p_\ell(\ol\part):= p_\ell \lrp{\frac{\part}{\part t_1},\frac12,\frac{\part}{\part t_2},\dots}$.}
\[
\lrp{ p_{k+4}(\tilde\part) -\frac12 \frac{\part^2}{\part t_1 \part t_{k+3}}} \tau_n\circ \tau_n=0,\]
of which the first equation reads
\[ \displaylines{
\lrp{\lrp{\frac{\part}{\part t_1}}^4 + 3\lrp{\frac{\part}{\part t_2}}^2  - 4\lrp{\frac{\part^2}{\part t_1 \part t_3}}}\log \tau_n+6 \lrp{\frac{\part^2}{\part t^2_1}\log \tau_n}^2=0,\hfill\cr\hfill k=0,1,2,\dots \;.}\]
\item[(iii)]  The \emph{standard Toda lattice}; i.e., the tridiagonal matrix
\begin{equation}\lab{5.16}
L(t):=S(t)\La S(t)^{-1}=
\begin{pmatrix}
\frac{\part}{\part t_1}\log \frac{\tau_1}{\tau_0}&\lrp{\frac{\tau_0\tau_2}{\tau^2_1}}^{1/2}&0\\
\lrp{\frac{\tau_0\tau_2}{\tau^2_1}}^{1/2}& \frac{\part}{\part t_1}\log \frac{\tau_2}{\tau_1}&
\lrp{\frac{\tau_1\tau_3}{\tau^2_2}}^{1/2}&\\
0&\lrp{\frac{\tau_1\tau_3}{\tau^2_2}}^{1/2}&\frac{\part}{\part t_1}\log \frac{\tau_3}{\tau_2}&\ddots\\&&\ddots&\ddots\end{pmatrix}
\end{equation}
satisfies the commuting equations\footnote{$(\;)_-$ means take the skew-symmetric part of $(\;)$ in the decomposition ``skew-symmetric" + ``lower-triangular."}  \pref{2.19}
\begin{equation}
\lab{5.17}
\frac{\part L}{\part t_k} = \lrb{ \frac12 (L^k)_-,L}.\end{equation}
\item[(iv)] \emph{Orthogonal polynomials}: The $n^{\rm th}$ degree polynomials $p_n(t;z)$ in $z$, depending on $t\in \C^\infty$, orthonormal with respect to the $t$-dependent inner product \pref{5.11}
\[
\lra{p_k(t;z) , p_\ell(t;z)}=\delta_{k\ell}\]
are eigenvectors of $L$, i.e., $(L(t)p(t;z))_n=zp_n(t;z)$, $n\geq 0$, and enjoy the following representations:
\begin{align*}
p_n(t;z):= (S(t)\chi(z))_n&= \frac{1}{\sqrt{\tau_n(t)\tau_{n+1}(t)}} \det \lrp{\begin{array}{c|c}
m_n & \begin{matrix} 1\\ z\\ \vdots\end{matrix}\\
\hline
\mu_{n,0} \cdots \mu_{n,n-1}& z^n\end{array}}
\\
&= z^n h^{-1/2}_n \frac{\tau_n(t-[z^{-1}])}{\tau_n(t)},\quad h_n:= \frac{\tau_{n+1}(t)}{\tau_n(t)}.\end{align*}
The functions $q_n(t;z):= z\int_{\R^n} \frac{p_n(t;u)}{z-u} \rho_t(u)du$ are ``dual eigenvectors" of $L$,  i.e., $(L(t)q(t;z))_n = zq_n(t;z)$, $n\geq 1$, and have the following $\tau$-function representation (see the remark after \pref{5.20}):
\begin{align}\label{5.18}
 q_n(t;z) := z\int_{\R^n} \frac{p_n(t;u)}{z-u} \rho_t (u)du&= \lrp{S^{\top-1}(t)\chi(z^{-1})}_n\\
&= \lrp{S(t)m_\infty (t)\chi(z^{-1})}_n\notag\\
&= z^{-n}h^{-1/2}_n \frac{\tau_{n+1}(t+[z^{-1}])}{\tau_n(t)}.\notag \end{align}
\item[(v)] \emph{Bilinear relations:} for all $n,m\geq 0$, and $a,b\in \C^\infty$, such that $a-b=t-t'$,
\begin{align}\label{5.19}
&\oint_{z=\infty} \tau_n (t-[z^{-1}])\tau_{m+1} (t'+[z^{-1}]) e^{\sum^\infty_1 a_iz^i} z^{n-m-1}\frac{dz}{2\pi i}\\
={}& \oint_{z=0} \tau_{n+1}(t+[z])\tau_m(t'-[z]) e^{\sum^\infty_1 b_iz^{-1}} z^{n-m-1} \frac{dz}{2\pi i}.\notag\end{align}

\end{itemize}}    \end{thm}

In the case $\beta=2$, the Virasoro expressions take on a particularly elegant form, namely, for $n\geq 0$,
\begin{align*}
\J^{(2)}_{k,n}(t)&=\sum_{i+j=k} : \J^{(1)}_{i,n}(t)\J^{(1)}_{j,n} (t):= J^{(2)}_k (t)+2nJ^{(1)}_k(t)+ n^2\delta_{0k}\\
\J^{(1)}_{k,n}(t)&=  J^{(1)}_k(t)+n\delta_{0k},\end{align*}
with\footnote{The expression $J^{(1)}_k=0$ for $k=0$.}
\begin{align} \label{5.20}
J^{(1)}_k &= \frac{\part}{\part t_k}+\frac12(-k) t_{-k},\\
J^{(2)}_k &= \sum_{i+j=k} \frac{\part^2}{\part t_i \part t_j} +\sum_{-i+j=k} it_i \frac{\part}{\part t_j}+\frac14
\sum_{-i-j=k} it_i jt_j.\notag\end{align}

\begin{remn}L. Haine and E. Horozov \cite{HH} have shown that these $\tau$-functions are highest-weight vectors in the context of Lie algebras.
\end{remn}

\begin{remn}
\emph{The vectors $p$ and $q$ are eigenvectors of $L$.}  Indeed, remembering $\chi(z)=(1,z,z^2,\ldots)^\top$, we have\[
\La \chi(z)=z\chi(z)\hensp{and}\La^\top \chi(z^{-1})=z\chi(z^{-1})-z e_1,\hensp{with} e_1=(1,0,0,\ldots)^\top.\]
Therefore, $p(z)=S \chi(z)$ and $q(z)=S^{\top-1}\chi(z^{-1})$ are eigenvectors, in the sense 
\begin{align*}
L p &=S\La S^{-1} S \chi(z)=z S_\chi (z)=zp,\\
L^\top q&= S^{\top-1}\La^\top S^\top S^{\top-1} \chi(z^{-1})\\
&= z S^{\top-1} \chi(z^{-1})-zS^{\top -1} e_1 = zq-zS^{\top-1} e_1.\end{align*}
Then, using $L=L^\top$, one is lead to
\[
((L-zI)p)_n = 0\hensp{for} n\geq 0\hensp{and} ((L-zI)q)_n =0\hensp{for }n\geq 1.\]
\end{remn}

\subsubsection*{Symmetric/symplectic matrix integrals and the Pfaff lattice} Consider an inner-product with a skew-symmetric weight $\rho(y,z)$,
\begin{equation}
\lab{5.21}
\lra{f,g}_t = \iint_{\R^2} f(y)g(z) e^{\sum^\infty_1 t_i(y^i+z^i)} \rho(y,z) dy\dz,\hensp{with}\rho(z,y)=-\rho(y,z).\end{equation}
Then, since
\[
\lra{f,g}_t = - \lra{g,f}_t,\]
the (semi-infinite) moment matrix, depending on $t=(t_1,t_2,\dots)$,
\[
m_n(t)=(\mu_{ij}(t))_{0\leq i,j\leq n-1} = \lrp{\lra{y^i,z^j}_t}_{0\leq i,j\leq n-1}\]
is skew-symmetric and the semi-infinite matrix $m_\infty$ evolves in $t$ according to the \emph{commuting vector  fields}
\begin{equation}\lab{5.22}
\frac{\part \mu_{ij}}{\mu t_k}=\mu_{i+k,j}+\mu_{i,j+k},\hensp{i.e.,} \frac{\part m_\infty}{\part t_k} = \La^k m_\infty + m_\infty \La^{\top k}.\end{equation}
It is well known that the determinant of an odd skew-symmetric matrix equals 0, whereas the determinant of an even skew-symmetric matrix is the square of a polynomial in the entries, the Pfaffian, with a sign specified below.
So
\begin{align*}
\det(m_{2n-1}(t))&=0\\
(\det m_{2n} (t))^{1/2}&= pf(m_{2n}(t)) = \frac{1}{n!} (dx_0\wedge dx_1\wedge\cdots\wedge dx_{2n-1})^{-1}\\
&\hspace*{1.5in} \lrp{\sum_{0\leq i<j \leq 2n-1}\mu_{ij}(t) dx_i\wedge dx_j}^n.\end{align*}
Define now the \emph{Pfaffian $\tau$-functions}:
\begin{equation}\lab{5.23}
\tau_{2n}(t):= pf\; m_{2n}(t).\end{equation}

Considering as a special skew-symmetric weight \pref{5.21}
\begin{equation}\lab{5.24}
\rho(y,z):=2D^\alpha \delta(y-z)\tilde\rho (y)\tilde \rho(z),\hensp{with}\alpha=\mp 1,\; \tilde\rho(y)=e^{-\tilde V(y)},\end{equation}
the inner-product \pref{5.21} becomes\footnote{$\vareps(y)=\text{sign}(y)$, and $\{f,g\}:= f'g-fg'$.  Also notice  that $\vareps' = 2\delta(x)$.} 
\begin{align*}
\lra{f,g}_t&= \iint_{\R^2} f(y)g(z) e^{\sum t_i(y^i+z^i)} 2D^\alpha \delta (y-z)\tilde \rho(y)\tilde\rho(z)\dy\dz\\
&=\bcs \displaystyle 
\iint_{\R^2} f(y)g(z) e^{\sum^\infty_1 t_i(y^i+z^i)}\vareps (y-z)\tilde\rho(y)\tilde\rho(z)\dy\dz&\hensp{for}\alpha=-1\\[11pt]  \displaystyle 
\int_\R\{f,g\}(y) e^{\sum^\infty_1 2t_i y^i} \tilde\rho(y)^2 \dy &\hensp{for} \alpha=+1,\ecs\end{align*}
and 
\begin{equation}\begin{split}
\lab{5.25}
&pf\lrp{\lra{y^i,z^j}_t}_{0\leq i,j\leq 2n-1}\\& \quad{} = 
\bcs \displaystyle
\frac{1}{(2n)!} \int_{\R^{2n}} |\Delta_{2n}(z)| \prod^{2n}_{k=1} e^{\sum^\infty_1 t_iz^i_k}\tilde \rho (z_k)dz_k&\\[11pt]
\displaystyle\qquad \qquad =\frac{1}{(2n)!} \int_{\cS_{2n}} e^{\Tr(-\tilde V(X)+\sum t_i X^i)} dX&\hbox{for }\alpha=-1,\\[11pt]
\displaystyle\frac{1}{n!} \int_{\R^n} |\Delta_n(z)|^4 \prod^n_{k=1} e^{\sum^\infty_1 2t_iz^i_k} \tilde\rho^2 (z_k)dz_k&\\[11pt]
\displaystyle\qquad\qquad = \frac{1}{n!} \int_{\cT_{2n}} e^{\Tr(-2\tilde V(X)+\sum 2 t_iX^i)} dX&\hbox{for }\alpha=+1.\ecs
\end{split}\end{equation} 
Setting
\[
\bcs
\tilde \rho (z)= \rho(z) I_E(z)&\hbox{for }\alpha=-1\\
\tilde\rho(z)= \rho^{1/2} (z)I_E(z),\; t\mapsto t/2&\hbox{for }\alpha=+1\ecs\]
in the identities \pref{5.25}, we are led to the identities between integrals and Pfaffians, which are spelled out in \cite{29} Theorem \ref{t5.3}:
\begin{thm}\lab{t5.3}
The integrals $I_n(t,c)$,
\[\displaylines{
I_n=\int_{E^n} |\Delta_n(z)|^\beta \prod^n_{k=1} \lrp{ e^{\sum^\infty_1 t_iz^i_k}\rho(z_k) dz_k}\hfill\cr\hfill
=\bcs \displaystyle n! pf\lrp{\iint_{E^2} y^i z^j \vareps (y-z) e^{\sum^\infty_1 t_k(y^k+z^k)} \rho(y) \rho(z)\dy\dz}_{0\leq i,j\leq n-1} \\  \displaystyle
 \hspace*{1.4in}= n!\tau_n(t,c),\qquad\qquad n\hensp{even for } \beta=1\\
n!pf\lrp{\int_E\{y^i,y^j\} e^{\sum^\infty_1 t_k y^k}\rho(y)\dy}_{0\leq i,j\leq 2n-1} =n!\tau_{2n}(t/2,c)\\
 \hspace*{2.64in} n\hbox{ arbitrary for }\beta=4\ecs}\]
 and the $\tau_n(t,c)$'s above satisfy the following equations:
{\setbox0\hbox{(iv)}\leftmargini=\wd0 \advance\leftmargini\labelsep
  \begin{itemize}
 \item[(i)] The \emph{Virasoro constraints}\footnote{Here the $a_i$'s and $b_i$'s are defined in the usual way, in terms of $\rho(z)$; namely $-\frac{\rho'}{\rho} = \frac{\sum b_i z^i}{\sum a_iz^i}$.} \pref{5.4} for $\beta=1,4$,
 \begin{equation}\lab{5.26}
 \lrp{-\sum^{2r}_1 c^{k+1}_i f(c_i) \frac{\part}{\part c_i} +\sum_{i\geq 0} \lrp{a_i\,\bj^{(2)}_{k+i,n} - b_i\, \bj^{(1)}_{k+i+1,n}}}I_n=0.\end{equation}
 \item[(ii)] The \emph{Pfaff-KP hierarchy}: (see footnote  11)
 \begin{equation} \begin{split}
 \lab{5.27}
 \lrp{ p_{k+4}(\tilde \part) - \frac12 \frac{\part^2}{\part t_1 \part t_{k+3}}}\tau_n\circ\tau_n = p_k(\tilde \part)\tau_{n+2}\circ \tau_{n-2} &\\
 n\hbox{ even, }k=0,1,2,\dots\ &\end{split}\end{equation}
 of which the first equation reads
 \[\displaylines{
 \lrp{\lrp{\frac{\part}{\part t_1}}^4+3\lrp{\frac{\part}{\part t_2}}^2-4 \frac{\part^2}{\part t_1 \part t_3}}\log\tau_n+6\lrp{\frac{\part^2}{\part t^2_1}\log \tau_n}^2\hfill\cr\hfill = 12\frac{\tau_{n-2}\tau_{n+2}}{\tau^2_n},\qquad n\hbox{ even.}}\]
 
 \item[(iii)] The \emph{Pfaff lattice:} The time-dependent matrix
 \begin{equation}\lab{5.28}
 L(t) = Q(t)\La Q(t)^{-1}\end{equation}
 satisfies the Hamiltonian commuting equations, as in \pref{2.4} 
 \[
 \frac{\part L}{\part t_i} = [-P_+(L^i),L],\qquad \hbox{(Pfaff lattice)}.\]
 \item[(iv)] \emph{Skew-orthogonal polynomials:} The vector of time-dependent polynomials $q(t;z):= (q_n(t;z))_{n\geq 0} = Q(t)\chi(z)$ in $z$ satisfy the eigenvalue problem
 \begin{equation}\lab{5.29}
 L(t) q(t,z) = zq(t,z)\end{equation}
 and enjoy the following representations:
 \begin{align*}
 q_{2n}(t;z) &= z^{2n} h^{-1}_{2n} \frac{\tau_{2n}(t-[z^{-1}])}{\tau_{2n(t)}},\quad h_{2n}=\frac{\tau_{2n+2}(t)}{\tau_{2n}(t)}\\
 q_{2n+1}(t;z) &= z^{2n} h^{-1/2}_{2n} \frac{1}{\tau_{2n}(t)} \lrp{z+\frac{\part}{\part t_1}} \tau_{2n} (t-[z^{-1}]).\end{align*}
 They are skew-orthogonal polynomials in $z$; i.e., 
 \[
 \lra{q_i (z;t),q_j(t;z)}_t = J_{ij}.\]
 \item[(v)] \emph{The bilinear identities:} For all $n,m\geq 0$, the $\tau_{2n}$'s satisfy the following biliear identity:
 \begin{equation}\lab{5.30}\begin{split}
& \oint_{z=\infty} \tau_{2n}(t-[z^{-1}])\tau_{2m+2}(t' + [z^{-1}]) e^{\sum^\infty_1(t_i-t'_i)z^i}z^{2n-2m-2} \frac{dz}{2\pi i}\\
 &\qquad{}+ \oint_{z=0} \tau_{2n+2}(t+[z])\tau_{2m}(t'-[z])e^{\sum^\infty_1(t'_i-t_i)z^{-i}} z^{2n-2m} \frac{dz}{2\pi i}=0.\end{split}\end{equation}
 \end{itemize}}   
 \end{thm}
 \subsection{Expressing $t$-partials in terms of boundary-partials}
 
 Given first-order linear operators $\cD_1,\cD_2,\cD_3$ in $c=(c_1,\dots,c_{2r})\in \R^{2r}$ and a function $F(t,c)$, with $t\in \C^\infty$, satisfying the following partial differential equations in $t$ and $c$:
 \begin{equation}\lab{5.31}
\cD_k F =\frac{\part F}{\part t_k} + \sum_{-1\leq j<k} \gamma_{kj} V_j(F)+\gamma_k +\delta_k t_1,\quad k=1,2,3,\dots,\end{equation}
with $V_j(F)$ nonlinear differential operators in $t_i$ of which the first few are given here:
\begin{equation}\lab{5.32}
V_j(F)=\sum_{i,i+j\geq 1} it_i \frac{\part F}{\part t_{i+j}} +\frac{\beta}{2} \delta_{2,j}\lrp{\frac{\part^2 F}{\part t^2_1}+\lrp{\frac{\part F}{\part t_1}}^2},\quad -1\leq j\leq 2.\end{equation}
In \pref{5.31} and \pref{5.32}, $\beta>0$, $\gamma_{ij},\gamma_k,\delta_k$ are arbitrary parameters; also $\delta_{2j}=0$ for $j\ne 2$ and $\delta_{2j}=1$ for $j=2$.  The claim is that the equations \pref{5.31} enable one to express all partial derivatives
\begin{equation}\lab{5.33}
\frac{\part^{i_1+ \cdots +i_k} F(t,c)}{\part t^{i_1}_1 \cdots \part t^{i_k}_k}\Big|_{\cL},\quad\hbox{along } \cL:= \{\hbox{all }t_i=0,\; c=(c_1,\dots,c_{2r})\hbox{ arbitrary}\},\end{equation}
uniquely in terms of polynomials in $\cD_{j_1}\cdots \cD_{j_r}F(0,c)$.  Indeed, the method consists of expressing $\frac{\part F}{\part t_k}\Big|_{t=0}$ in terms of $\cD_kf\big|_{t=0}$, using \pref{5.31}.  Second derivatives are obtained by acting on $\cD_k F$ with $\cD_\ell$, by noting that $\cD_\ell$ commutes with all $t$-derivatives, by using the equation for $\cD_\ell F$, and by setting in the end $t=0$:
\begin{equation}\lab{5.34}\begin{split}
\cD_\ell \cD_kF&= \cD_\ell \frac{\part F}{\part t_k} +\sum_{-1\leq j<k} \gamma_{kj} \cD_\ell (V_j(F))\\
&= \lrp{\frac{\part}{\part t_k}+\sum_{-1\leq j<k} \gamma_{kj} V_j} \cD_\ell(F),\hbox{ provided $V_j(F)$ does not}\\[-12pt]
&\hspace*{2.5in} \hbox{contain nonlinear terms}\\
&= \lrp{\frac{\part}{\part t_k}+\sum_{-1\leq j<k} \gamma_{kj} V_j} \lrp{\frac{\part F}{\part t_\ell}+\sum_{-1\leq j <\ell} \gamma_{\ell j}V_j(F)+\delta_\ell t_1}\\
&= \frac{\part^2 F}{\part t_k \part t_\ell}+\hbox{ lower-weight terms.}\end{split}\end{equation}
When the nonlinear term is present, it is taken care of as follows:
\begin{align}\label{5.35} 
\cD_\ell\lrp{\frac{\part F}{\part t_1}}^2&= 2\frac{\part F}{\part t_1} \cD_\ell \frac{\part F}{\part t_1}\\
&= 2\frac{\part F}{\part t_1} \frac{\part}{\part t_1}\cD_\ell F\nn\\
&= 2\frac{\part F}{\part t_1} \frac{\part}{\part t_1} \lrp{\frac{\part F}{\part t_\ell}+\sum_{-1\leq j<\ell} \gamma_{\ell j} V_j(F) +\gamma_\ell + \delta_\ell t_1};\nn\end{align}
higher derivatives are obtained in the same way.

\subsection{Using the KP-like equations}
Let
\[
\delta^\beta_{1,4}:= 2\lrp{\lrp{\frac{\beta}{2}}^{1/2}-\lrp{\frac{\beta}{2}}^{-1/2}}^2= \bcs
0&\hensp{for}\beta=2\\
1&\hensp{for}\beta=1,4.\ecs\]

From Theorems \ref{t5.2} and \ref{t5.3}, the integrals $I_n(t,c)$, depending on $\beta=2,1,4$, on $t=(t_1,t_2,\dots)$ and on the boundary points $c=(c_1,\dots,c_{2r})$ of $E$, relate to $\tau$-functions, as follows:
\begin{align}\label{5.36}
I_n(t,c)&= \int_{E_n}|\Delta_n(z)|^\beta \prod^n_{k=1} \lrp{e^{\sum^\infty_1 t_iz^i_k}\rho(z_k)dz_k}\\ \nn
&= \bcs
n!\tau_n(t,c),\phantom{/2} \quad n\hensp{arbitrary,}&\beta = 2\\
n!\tau_n(t,c),\phantom{/2}  \quad  n\hensp{even,}&\beta=1\\ 
n!\tau_{2n}(t/2,c)  \quad n\hensp{arbitrary,}&\beta=4.\ecs\end{align}
$I_n(t)$ refers to the integral \pref{5.36} over the full range.  It also follows that $\tau_n(t,c)$ satisfies the $\KP$-like equation
\begin{equation}\lab{5.37}
12\frac{\tau_{n-2}(t,c)\tau_{n+2}(t,c)}{\tau_n(t,c)^2}  \delta^\beta_{1,4}= (\KP)_t \log \tau_n(t,c),\quad\bcs
n&\hensp{arbitrary for} \beta=2\\
n&\hensp{even for}\beta=1,4,\ecs\end{equation}
where
\[
(\KP)_t F: = \lrp{\lrp{ \frac{\part}{\part t_1}}^4 +3\lrp{\frac{\part}{\part t_2}}^2 - 4\frac{\part^2 }{\part t_1 \part t_3}} F+6\lrp{\frac{\part^2 }{\part^2_1}F}^2.\]

$\beta=2,1$. Evaluating the left-hand side of \pref{5.35} (for $\beta=1$) yields, taking into account $P_n:= P_n(E)=I_n(0,c)/I_n(0)$,
\begin{align*}
\left.12\frac{\tau_{n-2}(t,c)\tau_{n+2}(t,c)}{\tau_n(t,c)^2}\right|_{t=0} &
=\left. 12\frac{(n!)^2}{(n-2)!(n+2)!} \frac{I_{n-2}(t,c) I_{n+2}(t,c)}{I_n(t,c)^2} \right|_{t=0}\\
&= 12 \frac{n(n-1)}{(n+1)(n+2)} \frac{I_{n-2}(0) I_{n+2}(0)}{I_n(0)^2} \frac{P_{n-2}P_{n+2}}{P^2_n}\\
&= 12 b^{(1)}_n \frac{P_{n-2}(E)P_{n+2}(E)}{P^2_n(E)},\end{align*}
with $b^{(1)}_n$ a constant.  Concerning the right-hand side of \pref{5.37}, it follows from Theorem \ref{t5.1} that $F_n(t;c)=\log I_n(t;c)$, as in  (5.36), satisfies Virasoro constraints.  As explained in \pref{5.31}--(5.35), we express
\[
\left.\frac{\part^4 F}{\part t^4_1}\right|_{t=0},\; \left.\frac{\part^2 F}{\part t^2_2}\right|_{t=0},\;
\left.\frac{\part^2 F}{\part t_1 \part t_3}\right|_{t=0},\; \left.\frac{\part^2 F}{\part t^2_1}\right|_{t=0},\quad 
F=\log I_n(t,c) \]
in terms of $\cD_k$ then $\cB_k$,  which are linear combinations of the $\cD_k$,which when substituted in the right-hand side of \pref{5.37}, i.e., in the KP-expressions, leads to the following theorems.

\subsubsection*{Hermitian, symmetric and symplectic Gaussian ensembles}
Given the disjoint union $E\subset \R$ and the weight $e^{-bz^2}$, the differential operators $\cB_k$ take on the form
\[
\cB_k = \sum^{2r}_1 c^{k+1}_i \frac{\part}{\part c_i}.\]
Also, define the \emph{invariant} polynomials\footnote{$Q(-2n,\beta,-\frac{a}{2},-\frac{b}{2})\Big|_{\beta=1} =Q(n,\beta,a,b)\big|_{\beta=4}$.}
\[
Q=12 b^2 n \lrp{n+1-\frac{2}{\beta}} ,\quad Q_2 = 4(1+\delta^\beta_{1,4})b\lrp{2n+\delta^\beta_{1,4}\lrp{1-\frac{2}{\beta}}}\]
and
\[
Q_1 = (2-\delta^\beta_{1,4})\frac{b^2}{\beta}.\]

\begin{thm}[Adler-van Moerbeke \cite{14}] \lab{t5.4}
The following probabilities for $(\beta=2,1,4)$
\[
P_n(E)=\frac{\int_{E_n}|\Delta_n(z)|^\beta \prod^n_{k=1} e^{-bz^2_k} dz_k}{\int_{\R^n}|\Delta_n(z)|^\beta \prod^n_{k=1}e^{-bz^2_k} dz_k},\]
satisfy the {\rm PDE}s $(F:= F_n=\log P_n)$:
\begin{align*}
&\delta^\beta_{1/4}Q\lrp{\frac{P_{n-{2\atop 1}} P_{n+{2\atop 1}}}{P^2_n}-1}\hensp{with index}
\bcs 2&\hbox{when $n$ is even and } \beta=1\\
1&\hbox{when $n$ is arbitrary and } \beta=4\ecs\\
&\qquad{}= \lrp{\cB^4_{-1} +(Q_2+6\cB^2_{-1} F)\cB^2_{-1}+4Q_1 (3\cB^2_0 - 4\cB_{-1} \cB_1 +6\cB_0)} F.\end{align*}
\end{thm}

\subsubsection*{Hermitian, symmetric and symplectic Laguerre ensembles} Given the disjoint union $E\subset \R^+$ and the weight $z^a e^{-bz}$, the $\cB_k$ take on the form
\[
\cB_k = \sum^{2r}_1 c^{k+2}_i \frac{\part}{\part c_i}.\]
Also define the polynomials, again respecting the duality (cf.\ footnote 16)
\begin{align*}
Q&= \bcs \frac34 n (n-1)(n+2a)(n+2a+1)&\hbox{for }\beta=1\\[7pt]
\frac32 n(2n+1)(2n+a)(2n+a-1)&\hbox{for }\beta=4,\ecs\\
Q_2&= \lrp{3\beta n^2 -\frac{a^2}{\beta} + 6an +4\lrp{1-\frac32}a+3} \delta^3_{1,4}+ (1-a^2)(1-\delta^\beta_{1,4}),\\
Q_1&= \lrp{\beta n^2 +2an+\lrp{1-\frac{\beta}{2}}a},\; Q_0=b(2-\delta^\beta_{1,4})\lrp{n+\frac{a}{\beta}},\\
Q_{-1}&=\frac{b^2}{\beta} (2-\delta^\beta_{1,4}).\end{align*}

\begin{thm}[Adler-van Moerbeke \cite{14}] \lab{t5.5} The following probabilities
\[
P_n(E)=\frac{\int_{E^n}|\Delta_n(z)|^\beta \prod^n_{k=1} z^a_k e^{-bz} dz_k}{\int_{\R^n_+ }|\Delta_n(z)|^\beta \prod^n_{k=1} z^a_k e^{-bz_k}dz_k} \]
satisfy the {\rm PDE}\footnote{With the index convention $\bcs 2&\hbox{when $n$ is even and $\beta=1$}\\
1&\hbox{when $n$ is arbitrary and $\beta=4$.}\ecs$}
$(F:=F_n=\log P_n)$
\begin{align*}
&\delta^\beta_{1,4}Q\lrp{\frac{P_{n-{2\atop 1}} P_{n+{2\atop 1}}}{P^2_n}-1}\\
&\quad{}=\Big(\cB^4_{-1} -2(\delta^\eta_{1,4} +1) \cB^3_{-1} \\&\qquad{}
+(Q_2+6\cB^2_{-1}
F-4(\delta^\beta_{1,4}+1)\cB_{-1}F)\cB^2_{-1}-3\delta^\beta_{1,4} (Q_1-\cB_{-1} F)\cB_{-1} \\
&\qquad{}+ Q_{-1}  (3\cB^2_0 -4\cB_1\cB_{-1}-2\cB_1) +Q_0(2\cB_0 \cB_{-1}-\cB_0)\Big) F.\end{align*}
\end{thm}

\subsubsection*{{\rm ODE}s, when $E$ has one boundary point}  Assume the set $E$ consists of one boundary point $c=x$, besides the boundary of the full range.  In that case the PDEs above lead to ODEs in~$x$:

(1) \emph{Gaussian $(n\times n)$ matrix ensemble} (for the function $\beta=2,1,4$):
\[
f_n(x)=\frac{d}{dx}\log P_n(\mathop{\max}_i \la_i\leq x)\]
satisfies
\begin{align*}
&\delta^{3}_{1,4}Q\lrp{\frac{P_{n-{2\atop 1}}P_{n+{2\atop 1}}}{P^2_n} -1} \\
&\qquad{}= f'''_n +6^{'2}_n  +\lrp{4\frac{b^2 x^2}{\beta}(\delta^\beta_{1,4} - 2) +Q_2}f'_n - 4\frac{b^2x}{\beta} (\delta^\beta_{1,4} -2)f_n.\end{align*}

(2) \emph{Laguerre ensemble} (for $\beta=2,1,4)$: all eigenvalues $\la_i$ satisfy $\la_i\geq 0$, and
\[
f_n(x)=x\frac{d}{dx}\log P_n(\mathop{\max}_i \la_i\leq x)\]
satisfies (with $f:=f_n(x)$)
\begin{align*}
&\delta^\beta_{1,4}Q \lrp{\frac{P_{n-{2\atop 1}} P_{n+{2\atop 1}}}{P^2_n}-1} -
\lrp{3\delta^\beta_{1,4} -\frac{b^2x^2}{\beta} (\delta^\beta_{1,4}-2)-Q_0 x - 3\delta^\beta_{1,4}Q_1}f\\
&= x^3f'''-(2\delta^\beta_{1,4}-1)x^2 f'' +6x^2 f^{'2}\\
&\quad{} - x\lrp{4(\delta^\beta_{1,4} +1) f-\frac{b^2x^2}{\beta}(\delta^\beta_{1,4}-2)-2Q_0 x -Q_2+ 2\delta^\beta_{1,4} +1} f'.\end{align*}

For $\beta=2$, $f_n(x)$ satisfies the third-order equation (of the so-called Chazy-type) with quadratic nonlinearity in $f'_n$ for the each of the ensembles, Gauss, Laguerre and Jacobi.  Then $f_n$ also satisfies an equation, which is the second-order in $f$ and quadratic in $f''$, which after some rescaling can be put in a canonical form.  Namely,
\begin{align*}
\hbox{Gauss}& \qquad g_n(z)=b^{-1/2} f_n(zb^{-1/2})+\frac23 nz,\\
\hbox{Laguerre}&\qquad g_n(z)=f_n(z)=\frac{b}{4} (2n+a)z+\frac{a^2}{4},\end{align*}
satisfies the respective canonical equations of Cosgrove and Cosgrove-Scoufis,
{\setbox0\hbox{(1}\leftmargini=\wd0 \advance\leftmargini\labelsep
 \begin{itemize}
\item $g''^{2}\! =\! -4g'^{3}\! +\! 4(zg' - g)^2\! +\!A_1 g' \!+\!A_2$ \hspace*{1.45in} (Painlev\'e IV)
\item $(zg'')^2 \!=\! (zg'-g)\lrp{-4g'^{2}\!+\!A_1(zg'-g) +A_2} \!+\!A_3 g' \!+\!A_4$ \enspace (Painlev\'e~V).
\end{itemize}}  \noindent 
For the ``Jacobi ensemble", L. Haine and J.-P. Semengue \cite{HS} have shown that the probabilities satisfy the Painlev\'e~VI equation. 

\section{The Spectrum of Coupled Random Matrices and the 2-Toda Lattice}
The purpose of this section is to derive a partial differential equation for the spectral gap probabilities for coupled-GUE-Hermitian random matrices \cite{15}.  The two-Toda lattice is the integrable deformation class of such coupled matrices and the tau-function theory, vertex theory and Virasoro symmetries play a crucial role in the derivation of the partial differential equations.  There is a huge literature on coupled random matrices, see  for example \cite{BZ}, as they play a huge role in matrix models, and string theory, as well as random matrix theory.

\subsection{Matrix Integrals and 2-Toda structure}
Consider the general weight $\rho(y,z)\dy\dz:= \rho_{t,s}(y,z)\dy\dz:= e^{V_{t,s}(y,z)} \dy\dz$ on $\R^2$, with $\rho_0=e^{V_0}$, where
\begin{equation}\lab{6.1}
V_{t,s} (y,z) := V_0(y,z)+\sum^\infty_1 t_iy^i - \sum^\infty_1 s_i z^i = \sum_{i,j\geq 1} c_{ij} y^i z^j+\sum^\infty_1 t_i y^i - \sum^\infty_1 s_i z^i,\end{equation}
with arbitrary $V_0$ and the inner product with regard to a subset $E\subset \R^2$
\begin{equation}
\lab{6.2}
\lra{f,g}_E=\int_E \dy\dz\rho_{t,s} (y,z) f(y)g(z).\end{equation}
Given the moment matrix (over $E$),
\begin{equation}\lab{6.3}
m_n (t,s,c)=: (\mu_{ij})_{0\leq ij \leq n-1}=\lrp{\lra{y^i,z^j}_E}_{0\leq i,j\leq n-1},\end{equation}
according to \cite{30}, the Borel decomposition of the semi-infinite matrix\footnote{$\cD_{k,\ell}$ $(k<\ell\in\Z)$ denotes the set of band matrices with zeros outside the strip $(k,\ell)$. $\La=(\delta_{j=i+1})_{i,j\geq 0}$.}
\begin{align*}
m_\infty&= S^{-1}_1 S_2\hensp{with} S_1\in \cD_{-\infty,0},\; S_2\in\cD_{0,\infty},\\
m_\infty(t,s,c)&= e^{\sum^\infty_1 t_n \La^n} m_\infty(0,0,c) e^{-\sum^\infty_1 s_n \La^{\top n}}\end{align*}
with $S_1$ having $1$'s on the diagonal and $S_2$ having $h_i$'s on the diagonal, leads to two strings $(p^{(1)} (y),p^{(2)}(z))$ of monic polynomials in one variable (dependent on $E$), constructed, in terms of the character $\ol\chi(z)=(z^n)_{n\in \Z,n\geq 0}$, as follows:
\begin{equation}\lab{6.4}
p^{(1)}(y) = :S_1\ol\chi(y),\quad p^{(2)}(z)=: (S^{-1}_2)^\top \ol\chi(z).\end{equation}
We call these two sequences \emph{bi-orthogonal polynomials}; in fact, according to \cite{30}, the Borel decomposition of $m_\infty=S^{-1}_1S_2$ above is equivalent to the ``orthogonality'' relations of the polynomials
\begin{equation}\lab{6.5}
\lra{p^{(1)}_n,p^{(2)}_m}_E=\delta_{n,m}h_n.\end{equation}
The matrices
\[
L_1:= S_1\La S^{-1}_1\hensp{and} L_2:= S_2\La^\top S^{-1}_2\]
interact with the vector of string orthogonal polynomials as follows:
\begin{equation}\lab{6.6}
L_1 p^{(1)}(y) = yp^{(1)}(y),\quad h L^\top_2 h^{-1} p^{(2)} (z)=zp^{(2)}(z).\end{equation}
Also define wave vectors $\Psi_1$ and $\Psi^\ast_2$ as follows:
\begin{alignat}{5} \label{6.7}
\Psi_1(z)&:= e^{\sum t_kz^k}p^{(1)}(z)\quad\hensp{and}\quad& \Psi^\ast_2(z)&:= e^{-\sum s_k z^{-k}} h^{-1} p^{(2)} (z^{-1})\\ \nn
&= e^{\sum t_k z^k}S_1\ol\chi(x)&&= e^{-\sum s_k z^{-k}}(S^{-1}_2)^\top \ol\chi(z^{-1}).\end{alignat}
As a function of $(t,s)$, the couple $L:= (L_1,L_2)$ satisfies the two-Toda lattice equations \pref{2.22}, and $\Psi_1$ and $\Psi^\ast_2$ satisfy \cite{15} (remember that $L$, $\Psi_1$ and $\Psi^\ast_2$ all depend on $E$)
\begin{align}
\label{6.8}
&\bcs \frac{\part}{\part t_n}\Psi = (L^n_0,0)_+ \Psi = ((L^n_1)_u(L^n_1)_u)\Psi\\[7pt]
\frac{\part}{\part s_n}\Psi = (0,L^n_2)_+\Psi = ((L^n_2)_\ell, (L^n_2)_\ell)\Psi\ecs\\[6pt] \nn
&\bcs \frac{\part}{\part t_n}\Psi^\ast = -((L^n_1,0)_+)^\top \Psi^\ast\\[7pt]
\frac{\part}{\part s_n}\Psi^\ast = -((0,L^n_2)_+)^\top \Psi^\ast,\ecs\end{align}
where in the above we introduced also $\Psi^\ast_1,\Psi_2$, $\Psi = (\Psi_1,\Psi_2)$, $\Psi^\ast=(\Psi^\ast_1,\Psi^\ast_2)$.
Moreover $(\tau^E_0:=1)$
\begin{equation}\lab{6.9}\begin{split}
n! \det m_n(t,s,c)&= \iint_{(u,v)\in E^n\subseteq \R^{2n}} \Delta_n (u)\Delta_n(v)\prod^n_{k=1}\lrp{
e^{V_{t,s}(u_k,v_k)}du_k dv_k}\\  
&= n!\det \lrp{E_n(t)m_\infty(0,0,c)E_n(-s)^\top}\\  
&= \prod^{n-1}_0 h_i (t,s,c)\\  
&= \tau^E_n(t,s,c),\end{split}\end{equation}
where $E_n(t):=$ (the first $n$ rows of $e^{\sum^\infty_1 t_n \La^n}$) is a matrix of Schur polynomials $p_n(t)$.  Also $\tau_n(t,s,c)$ is a $\tau$-function with regard to $t$ and $s$ and
\[
(L^k_1)_{nn} = \frac{\part}{\part t_k} \log \frac{\tau_{n+1}}{\tau_n},\quad (L^k_2)_{nn}=-\frac{\part}{\part s_k}
\log \frac{\tau_{n+1}}{\tau_n}.\]
In particular,
\[
L_1=\cdots + \frac{\part}{\part t_1}\log\frac{\tau_{n+1}}{\tau_n}+\La\hensp{and} L_2 = \frac{\tau_{n-1}\tau_{n+1}}{\tau^2_n}\La^{-1} -\frac{\part}{\part s_1}\log\frac{\tau_{n+1}}{\tau_n}+\cdots,\]
with the wave vectors parametrized by the $\tau$-functions as follows:
\begin{align}
\label{6.10}
\Psi_1(z) &= \lrp{ \frac{\tau_n(t-[z^{-1}],s)}{\tau_n(t,s)} e^{\sum^\infty_1 t_iz^i}z^n}_{n\geq 0},\\
\Psi^\ast_2(z)&= \lrp{\frac{\tau_n(t,s+[z])}{\tau_{n+1}(t,s)} e^{-\sum^\infty_1 s_iz^{-i}}z^{-n}}_{n\geq 0},\notag\end{align}
etc.\ for $\Psi^\ast_1,\Psi^\ast_2$.
Introducing the wave matrices
\begin{align}\label{6.11}
W_i & =S_i(t,s) e^{\xi_i(\La)}\\
 \xi_1(z)&=\sum^\infty_1 t_k z^k\qquad \xi_2(z)=\sum^\infty_1 s_k z^{-k},\nn\\
 \label{6.12}
 \Psi_i(t,s;z)& = W_i\ol\chi(z)=e^{\xi_i(z)} S_i \ol\chi(z),\\
 \Psi^\ast_i(t,s;z)&= (W^\top_i)^{-1}\ol\chi  (z^{-1})=e^{-\xi_i(z)}(S^\top_i)^{-1}\ol \chi(z^{-1}),\nn\end{align}
 from the relations \pref{6.8} the pair of matrices $W=(W_1,W_2)$ satisfies the bilinear relation (in the $\pm$ splitting of  \pref{2.20})
 \[
 (W(t,s)W(t',s')^{-1})_- = 0\]
 or equivalently,
 \begin{equation}\lab{6.13}
 W_1(t,s)W_1(t',s')^{-1}=W_2(t,s)W_2 (t',s')^{-1},\end{equation}
 from which one proves Proposition \ref{p6.1}; for details see \cite{31a, 31b}.
 \begin{prop}[bi-infinite and semi-infinite] \label{p6.1}
 The wave and adjoint wave functions satisfy, for all $m,n\in \Z$ (bi-infinite) and $m,n\geq 0$ (semi-infinite) and $t,s,t', s'\in\C^\infty$,
 \begin{equation} \lab{6.14}
 \oint_{z=\infty} \Psi_{1n}(t,s;z)\Psi^\ast_{1m} (t',s';z') \frac{dz}{2\pi i z} =\oint_{z=0} \Psi_{2n}(t,s;z)\Psi^\ast_{2m} (t',s';z') \frac{dz}{2\pi i z},\end{equation} \end{prop}
 which yields immediately from \pref{6.10}
 \begin{prop}\lab{p6.2}
 Two-Toda $\tau$-functions satisfy the following bilinear identities:
 \begin{align}\label{6.15}
 \oint_{z=\infty }&\tau_n (t-[z^{-1}],s)\tau_{m+1}(t'+[z^{-1}],s') e^{\sum^\infty_1 (t_i-t'_i)z^i} z^{n-m-1}\dz\\
 \nn&= \oint_{z=0}\tau_{n+1}(t,s-[z])\tau_m (t',s'+[z]) e^{\sum^\infty_1(s_i-s'_i)z^{-i}} z^{n-m-1}\dz,\end{align}
 or, expressed in terms of the Hirota symbol,\footnote{For the customary Hirota symbol $p(\part_t)f\circ g:= p\lrp{\frac{\part}{\part y}} f(t+y)g(t-y)\Big|_{y=0}$, with $\part_t = (\part_{t_1},\part_{t_2},\dots)$, $\tilde\part_t = (\part_{t_1},\frac12 \part_{t_2},\frac13 \part_{t_3},\dots)$.}
 \begin{align}
 \label{6.16}
 &\sum^\infty_{j=0} p_{m=n+j}(-2a) p_j (\tilde\part_t) e^{\sum^\infty_1\big(a_k\frac{\part}{\part t_k}+b_k\frac{\part}{\part s_k}\big)} \tau_{m+1}\circ\tau_n\\
 &\qquad {}= \sum^\infty_{j=0} p_{-m+n+j}(-2b) p_j(\tilde\part_s) e^{\sum^\infty_1 \big( a_k\frac{\part}{\part t_k}+b_k\frac{\part}{\part s_k}\big)} \tau_m\circ\tau_{n+1},\nn\end{align}
both for the bi-infinite $(n,m\in \Z)$ and the semi-infinite case ($n,m\in \Z$, $n,m\geq 0$). \end{prop}

\begin{proof}
\pref{6.15} follows at once from Proposition \ref{6.1} and the $\tau$-function representations \pref{6.10}, whereas \pref{6.16} follows from the shifts $t\mapsto t-a$, $t'\mapsto t'+a$, $s\mapsto s-b$, $s'\mapsto s'+b$, combined with the definition of the Hirota symbol.\end{proof}

This has as a direct consequence the following:
 Two-Toda $\tau$-functions $\tau(t,s)$ satisfy the KP-hierarchy in $t$ and $s$ separately, of which the first equation reads
\begin{equation}\label{6.17}
\lrp{\frac{\part}{\part t_1}}^4 \log\tau +6\lrp{\lrp{\frac{\part}{\part t_1}}^2\log\tau}^2+3\lrp{\frac{\part}{\part t_2}}^2 \log\tau-4\frac{\part^2}{\part t_1 \part t_3} \log\tau=0.\end{equation}
But they also satisfy the following identities \cite{15}:
\begin{thm}\lab{p6.3}
Two-Toda $\tau$-functions satisfy
\begin{equation*}
\lrc{\frac{\part^2 \log \tau_n}{\part t_1\part s_2},\frac{\part^2\log\tau_n}{\part t_1\part s_1}}_{t_1}+\lrc{
\frac{\part^2 \log\tau_n}{\part s_1\part t_2},\frac{\part^2 \log\tau_n}{\part t_1\part s_1}}_{s_1}=0 \end{equation*}
and
\begin{equation}\lab{6.18}
-\frac{\part}{\part s_1}\log\frac{\tau_{n+1}}{\tau_{n-1}} = \frac{ \frac{\part^2}{\part t_1 \part s_2}\log\tau_n}{\frac{\part^2 }{\part t_1 \part s_1}\log \tau_n},\quad \frac{\part}{\part t_1}\log \frac{\tau_{n+1}}{\tau_{n-1}} = \frac{\frac{\part^2}{\part s_1 \part t_2}\log\tau_n}{\frac{\part^2}{\part {s_1}\part t_1 }\log\tau_n}.\end{equation}\end{thm}

With the two-Toda lattice, we associate \emph{four} different vertex operators $\X_{ij}(\la,\mu)$ for $1\leq i,j\leq 2$; they map infinite vectors of $\tau$-functions into $\tau$-vectors.  The vertex operators $\X_{11}$ and $\X_{22}$ are basic vertex operators for Toda, and KP  as well, whereas $\X_{12}$ and $\X_{21}$ are vertex operators, native to two-Toda.  In particular, we construct for the semi-infinite two-Toda
\begin{equation}\lab{6.19}
\X_{12}(\mu,\la)=\La^{-1}\ol\chi (\la) X(-s,\la) X(t,\mu) \ol\chi(\mu),\end{equation}
with $\La$ the customary shift-operator $(\La v)_n=v_{n+1}$, and with
\[
X(t,\la):= e^{\sum^\infty_1 t_i\la^i} e^{-\sum^\infty_1 \la^{-i}\frac{1}{i}\frac{\part}{\part t_i}}.\]
Given a two-Toda lattice $\tau$-vector $\tau=(\tau_0,\tau_1,\dots)$, we have that $\tau+\X_{12}(y,z)\tau$ is another $\tau$-vector.  But more is true.  We show that the kernels $K_{12,n}(y,z)$, defined by the ratios $(\X_{12}\tau)_n/\tau_n$, have \emph{eigenfunction expansions} in terms of the eigenfunctions $\Psi$, reminiscent of the Christoffel-Darboux formula for orthogonal polynomials; to be precise \cite{15},
\begin{thm}\lab{t6.4}
We have for $\tau_n = \tau^\R_n$
\begin{equation}\lab{6.20}
K_{12,n} (y,z) := \frac{1}{\tau_n}\X_{12}(y,z)\tau_n =\sum_{0\leq j<n}\Psi^\ast_{2j} (z^{-1})\Psi_{1j}(y),\end{equation}
together iwth a Fredholm determinant-like formula
\begin{equation}\lab{6.21}
\det \lrp{  K_{12,n}(y_\alpha,z_\beta)}_{1\leq \alpha,\beta\leq k} = \frac{1}{\tau_n} \lrp{\sum^k_{\ell=1}\X_{12}(y_\ell,z_\ell)_\tau}_n.\end{equation}\end{thm}

\begin{cor}\lab{c6.5}
The vector of  Fredholm determinants equals 
\[
\det(I-\la K^E) =\frac{1}{\tau}e^{-\la \iint_E \dx\dy \rho_{t,s}(x,y) \X_{12}(x,y)}\tau\]
for the kernel $K^E=K_{12,n}(y,z)I_E(z)$, $\tau_n=\tau^\R_n$, with the measure $\rho_{t,s}(x,y) \dx\dy$.
\end{cor}

\begin{proof}
Putting the corresponding determinant obtained in Theorem \ref{t6.4} in the Fredholm formula below, we find for a subset of the form $E=E_1\times E_2\subset \R^2$, 
\begin{align*} &
(\det(I-\la K^E))_{n\in \Z}\\
&\quad{}=1+\sum^\infty_1 \frac{(-\la)^k}{k!} \int\cdots \int_{E^k}\det (K_{12,n}(x_i,y_i))_{1\leq i, j\leq k}\prod^k_1 (\rho_{t,s}(x_i,y_i)\dx_i\dy_i)\\
&\quad{}= \sum^\infty_0 \frac{(-\la)^k}{k!} \int\cdots\int_{E^k}\frac{1}{\tau} \lrp{\prod^k_1 \X_{12}(x_i,y_i)} \tau
\prod^k_1 (\rho_{t,s}(x_i,y_i)\dx_i\dy_i)\\
&\quad{}= \frac{1}{\tau}\sum^\infty_{k=0} \frac{1}{k!} \lrp{-\la \iint_E \X_{12}(x,y)\rho_{t,s}(x,y)\dx\dy}^k\tau\\
&\quad{}= \frac{1}{\tau} e^{-\la\iint_E \dx\dy \rho_{t,s}(x,y)\X_{12}(x,y)}\tau.\qedhere\end{align*}\end{proof}

\subsection{PDEs for the Gap Probabilities of Coupled Random Matrices}

Given the space of Hermitean matrices $\cH_N$, and given
\begin{align*}
\hbox{spectrum }M_1&= \{x_1,\dots,x_N\}\hbox{ and}\\
\hbox{spectrum }M_2&=\{y_1,\dots,y_N\},\hbox{ with } M_1,M_2\in\cH_N,\end{align*}
we define, for a set $E\subset \R^2$,
\[
\cH^2_{N,E} = \{(M_1,M_2)\in\cH^2_N\hensp{with all} (x_k,y_\ell)\in E\}.\]
Consider the product Haar measure $dM_1 dM_2$ on the product space $\cH^2_N$, with each $dM_i$, decomposed into its radial part and its angular part.  Also define the probability measure\footnote{$\ol\J^{(i)}_k = \J^{(i)}_k \big|_{t\mapsto -s}$, $i=1,2$.}
\begin{equation}\lab{6.22}
\frac{dM_1 dM_2 e^{\Tr V_{t,s}(M_1,M_2)} }{\iint_{\cH^2_{N^t}} dM_1 dM_2 e^{\Tr V_{t,s}(M_1,M_2)}}\end{equation}
and the Virasoro operators
\begin{align*}
\J^{(2)}_k &= \lrp{ J^{(2)}_{k,n}}_{n\in\Z} = \frac12 \lrp{J^{(2)}_k + (2n+k+1) J^{(1)}_k +n(n+1)J^{(0)}_k}_{n\in\Z},\\
\tilde \J^{(2)}_k&= \lrp{\tilde J^{(2)}_{k,n}}_{n\in\Z}=\frac12 \lrp{\tilde J^{(2)}_k +(2n+k+1)\tilde J^{(1)}_k +n(n+1)J^{(0)}_k}_{n\in\Z}.\end{align*}
Given the disjoint union
\begin{equation}\lab{6.23}
E=E_1\times E_2:= \bigcup^r_{i=1} [a_{2i-1} ,a_{2i}]\times \bigcup^s_{i=1} [b_{2i-1},b_{2i}]\subset \R^2,\end{equation}
define the following integral:
\begin{equation}\lab{6.24}
\U_E:= \iint_E \X_{12}(x,y) \rho_0(x,y)\dx\dy,\end{equation}
of the vertex operator $\X_{12}$, defined in \pref{6.19}.

This brings us to the following theorems:
\begin{thm}\lab{t6.6}
Given a set $E$, as in \pref{6.23},
the probability 
\begin{align}\label{6.25}
&P(\hbox{all $M_1$-eigenvalues $\in E_1$ and all $M_2$-eigenvalues }\in E_2)\\
&\qquad{}= \frac{\iint_{\cH^2_{n,E}} dM_1 dM_2 e^{\Tr V_{t,s}(M_1,M_2)}}{\iint_{\cH^2_n} dM_1 dM_2 e^{\Tr V_{t,s}(M_1,M_2)}}=: \frac{\tau^E_n}{\tau_n}\nn\end{align}
is a ratio of two $\tau$-functions $\tau^E_n$ and $\tau_n$, such that
\[
\tau^E_n=((\U_E)^n\tau)_n.\]
Moreover, $\tau_n$ and $\tau^E_n$ satisfy the partial differential equations, labeled for $k\geq -1$, and when all $c_{ij}=0$, but $c_{11}=: c$, we find
\begin{align}\label{6.26}
\lrp{-\sum^r_{i=1} a^{k+1}_i \frac{\part}{\part a_i} + J^{(2)}_{k,n}}\tau^E_n +cp_{k+n} (\tilde \part_t) p_n(-\tilde \part_s) \tau^E_1\circ \tau^E_{n-1}&=0,\\ \nn
\lrp{-\sum^s_{i=1} b^{k+1}_i \frac{\part}{\part b_i} + \tilde J^{(2)}_{k,n}}\tau^E_n+cp_n (\tilde\part_t) p_{k+n}(-\tilde\part_s) \tau^E_1\circ \tau^E_{n-1}&=0.\end{align}
\end{thm}

\begin{remn}
Whenever some $a_i$ or $b_i=\infty$, we must interpret  $a^{k+1}_i \frac{\part}{\part a_i}$ or $b^{k+1}_i \frac{\part}{\part b_i}\equiv 0$; in particular, $\tau_n$ satisfies the same equations, but without the boundary terms.
\end{remn}

The above formula \pref{6.26} depends on many results; namely,
setting
\begin{equation}\lab{6.27}
\V_k := -b^{k+1}\frac{\part}{\part b}-a^{k+1} \frac{\part}{\part a} +\J^{(2)}_k +\sum_{i,j\geq 1} ic_{ij}\frac{\part}{\part c_{i+k,j}},\end{equation}

\begin{thm}\lab{t6.7}
For all $k\geq -1$ and $n\geq 1$,
\[
[\V_k,(\U_E)^n] = 0,\]
with the vector $\J^{(2)}_k$ forming a Virasoro algebra of central charge $c=-2$:
\[
\lrb{\J^{(2)}_k,\J^{(2)}_\ell} = (k-\ell) \J^{(2)}_{k+\ell}+(-2)\lrp{\frac{k^3-k}{12}} \delta_{k,-\ell}\]
and the remarkable identity
\begin{equation}\lab{6.28}
\frac{\part\tau^E_n}{\part c_{\alpha\beta}} = \tau^E_n \sum^{n-1}_{i=0} (L^\alpha_1 L^\beta_2)_{ii} = p_{\alpha+n-1} (\tilde \part_t) p_{\beta+n-1} (-\tilde\part_s) \tau^E_1 \circ \tau^E_{n-1}.\end{equation}
\end{thm}

To prove the last formula of \pref{6.25} first observe that
\[
\int_{\cU(n)} dU e^{c\Tr xU yU^\top} = \frac{(2\pi)^{\frac{n(n-1)}{2}} \det(e^{cx_i y_j})_{1\leq i,j\leq n}}{n! c\frac{n(n-1)}{2}\Delta (x)\Delta(y)}\]
which implies that for $E=E_1\times E_2\subset \R^2$, the following holds:
\begin{align}
\label{6.29}
&\iint_{\cH^2_{N,E}} e^{c \Tr(M_1M_2) } e^{\Tr \sum^\infty_1 (t_i M^i_1 - s_i M^i_2)} dM_1 \, dM_2\\
&\qquad {}= \iint_{E^N} \prod^N_{k=1} \lrp{ dx_k \dy_k e^{\sum^\infty_{i=1}(t_k x^i_k -s_i y^i_k)+cx_k y_k}}\Delta_N(x)\Delta_N(y).\nn\end{align}

We now can prove \cite{15}:
\begin{prop}\lab{6.8}
For $E=E_1\times E_2 \subset \R^{2}$, we have
\begin{equation}\lab{6.30}
\tau^E_n = \lrp{(\cU_E)^n \tau}_n.\end{equation} \end{prop}

\begin{proof} In what follows, we use the monic bi-orthogonal polynomials \pref{6.4} $p^{(1)}_i,p^{(2)}_j$, defined by $\rho_{t,s}(x,y)$ on $\R^2$; therefore the $h_i(t,s,c)$ are the $\R^2$ inner products.  We first compute, using \pref{6.9} for $E=\R^2$, and remembering the notation \pref{6.24} and formulae \pref{6.9} and \pref{6.7},
\begin{align*}
\frac{\tau^E}{\tau_n}
&= \lrp{\prod^{n-1}_0 h^{-1}_i}\iint_{E_n} \prod^n_{k=1} (\dx_k \dy_k \rho_{t,s}(x_k,y_k)) \Delta_n (x)\Delta_n(y)\\
&=\lrp{\prod^{n-1}_0 h^{-1}_i } \iint_{E_n} \prod^n_{k=1} (\dx_k \dy_k \rho_{t,s} (x_k,y_k))
\det(\rho^{(1)}_{i-1} (x_j))_{1\leq i,j\leq n} \det \lrp{\rho^{(2)}_{i-1}(y_j)}_{1\leq i,j\leq n}\\
&= \iint_{E^n} \prod^n_{k=1}(\dx_k\dy_k\rho_0(x_k,y_k)) \det\lrp{e^{\sum t_i x^i_k}
\sum^n_{i=1} \rho^{(1)}_{i-1} (x_k) h^{-1}_{i-1}\rho^{(2)}_{i-1}(y_\ell)e^{-\sum s_i y^i_\ell}}_{1\leq k,\ell\leq n}\\
& = \iint_{E^n} \sum^n_{k=1} (\dx_k\dy_k\rho_0(x_k,y_k)) \det\lrp{\sum_{0\leq i\leq n-1} \Psi_{1i} (x_k)\Psi^\ast_{2i} (y^{-1}_\ell)}_{1\leq k,\ell\leq n}\\
&=
\iint_{E^n} \prod^n_{k=1} (\rho_0 \dx_k\dy_k) \det(K_n(x_k,y_\ell))_{1\leq k,\ell\leq n}\\
&= \iint_{E^n}\prod^n_{k=1}(\rho_0\dx_k\dy_k) \lrp{\frac1\tau \prod^n_{k=1} \X_{12}(x_k,y_k)\tau}_n \hbox{ using \pref{6.21}}\\
&= \lrp{\frac1\tau\lrp{ \iint_E \X_{12}(x,y)\rho_0(x,y)\dx\dy}^n\tau }_n,
 \end{align*}
 establishing \pref{6.30}.
 \end{proof}
 
 To give the next result, we define differential operators $\cA_k$, $\cB_k$ of \emph{weight} $k$, in terms of the coupling constant $c$, and the boundary of the set
 \begin{equation}\lab{6.31}
 E =E_1\times E_2:= \bigcup^r_{i=1} [a_{2i-1},a_{2i}] \times \bigcup^s_{i=1} [b_{2i-1},b_{2i}]\subset \R^2;\end{equation}\begin{alignat*}{5}
 \cA_1&= \frac{1}{c^2 -1} \lrp{\sum^r_1 \frac{\part}{\part a_j}+c \sum^s_1\frac{\part}{\part b_j}},&
 \quad \cB_1 &=\frac{1}{1-c^2} \lrp{c\sum^r_1 \frac{\part}{\part a_j}+\sum^s_1 \frac{\part}{\part b_j}} ;&\\
 \cA_2&=\sum^r_{j=1} a_j \frac{\part}{\part a_j}-c\frac{\part}{\part c},& \qquad \cB_2 &= \sum^s_{j=1}b_j \frac{\part}{\part b_j}-c\frac{\part}{\part c};\end{alignat*}
they form a Lie algebra parametrized by $c$:
\begin{alignat}{7}
\label{6.32}
{[}\cA_1,\cB_1] &= 0, \quad &[\cA_1,\cA_2]&= \frac{1+c^2}{1-c^2}\cA_1, \quad &[\cA_2,\cB_1]&= \frac{2c}{1-c^2}\cA_1&\\ \nn
{[}\cA_2,\cB_2]&=0,\quad &[\cA_1,\cB_2]&=\frac{-2c}{1-c^2}\cB_1,\quad &[\cB_1,\cB_2]&= \frac{1+c^2}{1-c^2}\cB_1.&\end{alignat}

The following theorem deals with the joint distribution \pref{6.22}, with 
\begin{equation}\begin{split} \lab{6.33}
 V_{t,s}(M_1,M_2)& = -\frac12 (M^2_1+M^2_2)+cM,M_2 ,\\
P_n (E)&:= P\hbox{(all ($M_1$-eigenvalues) $\in E_1$, all ($M_2$-eigenvalues) }\in E_2),\end{split}\end{equation}
and leads to a formula in \cite{15}, which is the ``mirror image" of Theorem \ref{6.3}.

\begin{thm}[Gaussian probability]\lab{t6.9}
The statistics \pref{6.33} satisfy the $n$-independent nonlinear third-order partial differential equation\footnote{Using the following relation for non-commutative operators $X$ and $Y$,
\[
XY\log f=\frac{1}{f^2}(f XYf-Xf Yf).\]} 
$(F_n:= \frac1n \log P_n(E))$:
\begin{equation}\lab{6.34}
\lrc{ \cB_2\cA_1F_n,\;\cB_1\cA_1F_n+\frac{c}{c^2-1}}_{\cA_1 } - \lrc{ \cA_2 \cB_1 F_n,\; \cA_1\cB_1F_n+\frac{c}{c^2-1}}_{\cB_1}=0.\end{equation}\end{thm}

\begin{rem}
Since the equation above for the joint statistics is independent of the size $n$, the same joint statistics for infinite coupled ensembles should presumably be given by the same partial differential equation.
\end{rem}
\begin{rem} For $E=E_1\times E_2:= (-\infty,a]\times (-\infty,b]$, equation \pref{6.34} takes on the following form: Upon introducing the new variables $x:= -a+cb$, $y:= -ac+b$, the differential operators $\cA_1$ and $\cB_1$ take on the simple form $\cA_1=\part/\part x$, $\cB_1 = \part/\part y$ and \pref{6.34} becomes
\[
\frac{\part}{\part x}\lrp{ \frac{(c^2-1)^2 \frac{\part^2 F_n}{\part x\part c}+2cx+(1+c^2)y}{(c^2-1)\frac{\part^2 F_n}{\part x\part y}+c}}=\frac{\part}{\part y}
\lrp{\frac{(c^2-1)^2 \frac{\part^2 F_n}{\part y\part c}+ 2cy-(1+c^2)x}{(c^2-1)\frac{\part^2 F_n}{\part y\part x}+c}}.\]
We sketch the proof  of \pref{6.34}.

\begin{proof} From \pref{6.31}, it clearly follows that
\[
P_n(E)=\left.\frac{\tau^E_n(t,s,c_{ij})}{\tau^{\R^2}_n (t,s,c_{ij})}\right|_{\cL},\]
where $\tau^E_n$ is an integral over $E^n\subset \R^{2n}$, i.e., $(x,y)\in E^n_1 \times E^n_2=E^n$,
\begin{align}
\label{6.35}
\tau^E_n(t,s,c_{ij})&=\iint_{E^n}\dx\dy \Delta_n(x)\Delta_n(y)\\
&\quad{}\cdot \prod^n_{k=1} e^{-\frac12 (x^2_k+y^2_k-2cx_k y_k)+\sum^\infty_{i=1}(t_i x^i_k -s_i y^i_k)+\sum_{{i,j\geq 1\atop (i,j)\ne (1,1)}} c_{ij} x^i_k y^j_k},\nn\end{align}
and where $\cL$ denotes the locus 
\[
\cL=\{t_i=s_i=0,\; c_{11}=c\hensp{and all other} c_{ij}=0\}.\]
We need to write down the Virasoro constraints \pref{6.26}:
\begin{alignat}{5}\label{6.36}
 \frac{\part}{\part t_1}\log \tau_n\Big|_{\cL}&= \cA_1 \log\tau_n\big|_{\cL} ,&\quad  
\frac{\part}{\part s_1 }\log \tau_n\Big|_{\cL}&= \cB_1 \log\tau_n\big|_{\cL} &\\
\nn
 \frac{\part}{\part t_2}\log \tau_n\Big|_{\cL}&= -\cA_2 \log \tau_n\big|_{\cL}+ \frac{n(n+1)}{n},&\quad 
 \frac{\part}{\part s_2} \log\tau_n\Big|_{\cL}&= \cB_2\log \tau_n\big|_{\cL}-
\frac{n(n+1)}{2},\\
\label{6.37}
 \frac{\part^2}{\part t_1\part s_1}\log\tau_n\Big|_{\cL} &= \cB_1\cA_1\log\tau_n+\frac{nc}{c^2-1}&&\\
\label{6.38}
 \frac{\part^2}{\part t_1 \part s_2}\log\tau_n\Big|_{\cL}&= \cB_2\cA_1\log\tau n.&&\end{alignat}
 Setting \pref{6.36}, \pref{6.37}, \pref{6.38} into the formula of Propostion \ref{p6.3} one is led to an expression for $\cB_1 \log \frac{\tau_{n+1}}{\tau_{n-1}}$ and a dual expression for $\cA_1 \log\frac{\tau_{n+1}}{\tau_{n-1}}$:
 \begin{align}
 \label{6.39}
 - \cA_1\log\frac{\tau_{n+1}}{\tau_{n-1}}& = \frac{\cA_2 \cB_1\log\tau_n}{\cA_1 \cB_2\log\tau_n+\frac{nc}{c^2-1}}\\
 \nn
 -\cB_1\log\frac{\tau_{n+1}}{\tau_{n-1}} &= \frac{\cB_2 \cA_1 \log\tau_n}{ \cB_2 \cA_1 \log\tau_n+\frac{nc}{c^2-1}}.\end{align}
 
 Upon taking $\cA_1$ of the second expression, subtracting from it $\cB_1$ of the first one and using $[\cA_1,\cB_1]=0$, one finds the following identity:
 \[
 \cA_1 \frac{\cB_2\cA_1\log\tau_n}{\cB_1\cA_1\log\tau_n +\frac{nc}{c^2-1}} - \cB_1 \frac{\cA_2\cB_1\log\tau_n}{\cA_1\cB_1\log\tau_n +\frac{nc}{c^2-1}} =0.\]
 This difference amounts to the equality of two Wronskians $\big(G_n:=\frac12 \log\tau_n\big)$:
 \begin{equation}\lab{6.40}
 \lrc{\cB_2\cA_1 G_n,\; \cB_1\cA_1G_n+\frac{c}{c^2-1}}_{\cA_1} = \lrc{\cA_2 \cB_1G_n,\; \cA_1\cB_1 G_n+\frac{c}{c^2-1}}_{\cB_1}.\end{equation}
 Because of the fact that
 \[
 \log P_n(E)=\log (\tau_n(E)/\tau_n(\R^2)) = \log\tau_n (E)-\log\tau_n(\R^2),\]
 together with the fact that $\cA_1\tau_n(\R^2)= \cB_1\tau_n(\R^2)=0$, we have that $F_n(E):= \frac{1}{p} \log P_n(E)$ satisfies \pref{6.40} as well, thus leading to \pref{6.34}.
 \end{proof}
\end{rem}
 \providecommand{\bysame}{\leavevmode\hbox to3em{\hrulefill}\thinspace}
\providecommand{\MR}{\relax\ifhmode\unskip\space\fi MR }
\providecommand{\MRhref}[2]{%
  \href{http://www.ams.org/mathscinet-getitem?mr=#1}{#2}
}
\providecommand{\href}[2]{#2}

\end{document}